\documentclass[11pt]{article}
\usepackage{amsmath,amsfonts,amssymb,amsthm}
\allowdisplaybreaks[1]
\usepackage{graphicx,color,slashed}
\usepackage[english]{babel}
\usepackage{url}
\usepackage{booktabs}
\usepackage{comment}
\usepackage{fullpage}
\usepackage{xcolor}
\usepackage{bm}
\usepackage{fancyhdr, dsfont}
\usepackage{float}
\usepackage{url}
\usepackage{setspace}
\usepackage[nosort]{cite}
\usepackage{mathtools}
\usepackage{color}
\usepackage{hyperref}
\definecolor{darkred}{rgb}{0.8,0.1,0.1}
\hypersetup{colorlinks=true, linkcolor=darkred, urlcolor=darkred,citecolor=blue, linktoc=page}

\font\tenshuffle=shuffle10 \font\sevenshuffle=shuffle7 \font\fiveshuffle=shuffle7 at 5pt
\def\shuffle{{%
\def\Dshuffle{\mathbin{\hbox{\tenshuffle\char'001}}}%
\def\Sshuffle{\mathbin{\hbox{\sevenshuffle\char'001}}}%
\def\SSshuffle{\mathbin{\hbox{\fiveshuffle\char'001}}}%
\mathchoice{\Dshuffle}{\Dshuffle}{\Sshuffle}{\SSshuffle}}}

\newcommand{\CC}{\mathbb C}

\newcommand{\MM}{\mathbb M}

\newcommand{\GG}{\mathbb G}
\newcommand{\HH}{\mathbb H}
\newcommand{\FF}{\mathbb F}

\newcommand{\QQ}{\mathbb Q}

\newcommand{\YY}{\mathbb Y}
\newcommand{\BB}{\mathbb B}

\newcommand{\dd}{\mathrm{d}}

\newcommand{\mot}{\mathfrak m}
\newcommand{\dR}{\mathfrak{dr}}
\newcommand{\sv}{\mathrm{sv}}

\newcommand{\bea}{\begin{eqnarray}}
\newcommand{\eea}{\end{eqnarray}}
\def\beq{\begin{equation}}
\def\eeq{\end{equation}}

\newtheorem{theorem}{Theorem}[section]

\newtheorem{lemma}[theorem]{Lemma}
\newtheorem{prop}[theorem]{Proposition}
\newtheorem{remark}[theorem]{Remark}
\newtheorem{definition}[theorem]{Definition}

\begin{document}
{\flushright  
 UUITP--09/25\\}

\begin{center}

\vspace{5mm}

{\bf {\Large \sc 
Deriving motivic coactions and single-valued maps \\[2mm] at genus zero from zeta generators }}

\renewcommand{\thefootnote}{\fnsymbol{footnote}}

\vspace{8mm}
\normalsize
{\large  Hadleigh Frost$^1$,
Martijn Hidding$^2$,
Deepak Kamlesh$^3$, \\ \vskip 0.45 em
Carlos Rodriguez$^{4,5}$, 
Oliver Schlotterer$^{4,6}$,
and Bram Verbeek$^{4}$} \\ \vskip 0.45 em

\vspace{8mm}
${}^1${\it Institute for Advanced Study, 1 Einstein Dr, Princeton, NJ 08540, U.S.A.}\footnote{\href{mailto:frost@ias.edu}{frost@ias.edu}}
\vskip 0.5 em
${}^2${\it Institute for Theoretical Physics, ETH Z\"urich, 8093 Z\"urich, Switzerland}
\vskip 0.5 em
${}^3${\it Max Planck Institute for Mathematics, 53111 Bonn, Germany}\,\footnote{\href{mailto:kamlesh@mpim-bonn.mpg.de}{kamlesh@mpim-bonn.mpg.de}}
\vskip 0.5 em
${}^4${\it Department of Physics and Astronomy, Uppsala University, 75108 Uppsala, Sweden}
\vskip 0.5 em
${}^5${\it Max Planck Institute for Mathematics in the Sciences, 04103 Leipzig, Germany}\,\footnote{\href{mailto:carlos.rodriguez@mis.mpg.de}{carlos.rodriguez@mis.mpg.de}}
\vskip 0.5 em
${}^6${\it Department of Mathematics, Centre for Geometry and Physics, \\ Uppsala University, Box 480, 75106 Uppsala, Sweden}\,\footnote{\href{mailto:oliver.schlotterer@physics.uu.se}{oliver.schlotterer@physics.uu.se}}
\vskip 0.5 em

\renewcommand{\thefootnote}{\arabic{footnote}}

\vspace{8mm}

\hrule

\vspace{8mm}

\begin{tabular}{p{150mm}}

Multiple polylogarithms are equipped with rich algebraic structures including the motivic coaction and the single-valued map which both found fruitful applications in high-energy physics. In recent work \href{https://arxiv.org/pdf/2312.00697}{arXiv:2312.00697}, the current authors presented a conjectural reformulation of the motivic coaction and the single-valued map via zeta generators, certain operations on non-commuting variables in suitable generating series of multiple polylogarithms. In this work, the conjectures of the reference will be proven for multiple polylogarithms that depend on any number of variables on the Riemann sphere.

\end{tabular}

\vspace{3mm}

\hrule

\end{center}

\thispagestyle{empty}

\newpage
\setcounter{page}{1}

\numberwithin{equation}{section}

\setcounter{tocdepth}{2}
\tableofcontents


\section{Introduction}
\label{sec:1}

Multiple polylogarithms (MPLs) \cite{GONCHAROV1995197, Goncharov:1998kja, Remiddi:1999ew, Goncharov:2001iea, Vollinga:2004sn} complete the space of rational functions in one or more complex variables $z_1,z_2,\ldots,z_n$ to close under integration over any of the $z_i$ \cite{Brown:2009qja}. Hence, MPLs find a wealth of applications in any discipline where integrations over the Riemann sphere arise. MPLs are central to perturbative computations in quantum field theory \cite{Duhr:2014woa, Henn:2014qga, Abreu:2022mfk, Weinzierl:2022}, string theory \cite{Berkovits:2022ivl, Mafra:2022wml} and adjacent areas of high-energy physics. Moreover, MPLs fruitfully connect algebraic geometry and number theory due to their study in the light of motivic periods and the appearance of multiple zeta values (MZVs) \cite{Jianqiang, GilFresan} as special values of MPLs. Many of the applications of MPLs within mathematics and to physics rely on the motivic coaction \cite{Goncharov:2001iea, Goncharov:2005sla, BrownTate, Duhr:2012fh, Brown2014MotivicPA}, and the single-valued map \cite{Schnetz:2013hqa, Brown:2013gia, Brown:2018omk}.

However, the construction of MPLs from iterated integrals of rational functions limits their application to integration on the Riemann sphere as opposed to Riemann surfaces of genus $\geq 1$ or higher-dimensional varieties. Advances to extend MPLs to the torus --- so-called elliptic polylogarithms \cite{Lev, LevRac, BrownLev, Broedel:2014vla, Broedel:2017kkb, Bourjaily:2022bwx, BEFZ:2307, Berkovits:2022ivl} --- and to higher-genus Riemann surfaces \cite{Enriquez:2011, BEFZ:2110, Ichikawa:2022qfx, BEFZ:2212, DHoker:2023vax, Baune:2024biq, DHoker:2024ozn, Baune:2024ber, DHoker:2025szl, DHoker:2025dhv, EZ:2025} became a vibrant interdisciplinary research area that stimulated numerous collaborations between mathematicians and physicists. The significance of the motivic coaction and single-valued map of MPLs for quantum field theory and string theory, see e.g.\ \cite{Duhr:2012fh, Schnetz:2013hqa, Cartier:1988,  Schlotterer:2012ny, Drummond:2013vz, Abreu:2014cla, Brown:2015fyf, Panzer:2016snt, Abreu:2017enx, Abreu:2017mtm, Schnetz:2017bko, Caron-Huot:2019bsq, Abreu:2019xep, Tapuskovic:2019cpr, Gurdogan:2020ppd, Britto:2021prf, Dixon:2021tdw, Borinsky:2022lds, Dixon:2022xqh, Dixon:2023kop, Dixon:2025zwj} and \cite{Schlotterer:2012ny, Dixon:2012yy, Drummond:2012bg, DelDuca:2013lma, Stieberger:2013wea, Stieberger:2014hba, Zerbini:2015rss, DHoker:2015wxz, DelDuca:2016lad, Broedel:2016kls, Broedel:2018izr, Schlotterer:2018abc, Vanhove:2018elu, Brown:2019wna, DelDuca:2019tur, Gerken:2020xfv, Alday:2022xwz, Alday:2023jdk, Fardelli:2023fyq, Duhr:2023bku, Alday:2024ksp, Alday:2024rjs}, respectively, provides tremendous motivation in high-energy physics to make similar structures accessible for elliptic and higher-genus polylogarithms.

With this long-term motivation in mind, the present authors proposed a reformulation~\cite{Frost:2023stm} of the motivic coaction and the single-valued map of MPLs in any number of complex variables. The proposed formulas use zeta generators, which are free generators of the motivic Lie algebra\footnote{Formally, the motivic Lie algebra we consider is the graded Lie algebra of the pro-unipotent part of the motivic Galois group of the category of mixed Tate motives unramified over $\mathbb{Z}$, denoted by $\mathcal{G}_U$ in \cite{BrownTate}. For our purposes, it suffices to regard this as a free Lie algebra generated by the zeta generators, which are labeled by the odd numbers $3,5,7,\ldots$.} \cite{DG:2005, BrownTate, Brown:Anatomy, Brown:depth3, LI202038}, to recast the expressions in the literature for motivic coactions \cite{Goncharov:2001iea, Goncharov:2005sla, BrownTate, Ihara1989TheGR,  Brown:2011ik, Brown:2019jng} and single-valued MPLs \cite{DelDuca:2016lad, Broedel:2016kls, svpolylog,  Charlton:2021uhu} into a more genus-agnostic form. Indeed, by the tight interplay between zeta generators at genus zero and genus one \cite{Dorigoni:2024iyt}, Brown's single-valued iterated Eisenstein integrals \cite{Brown:mmv, Brown:2017qwo, Brown:2017qwo2} were recently generated from certain series in genus-one zeta generators
\cite{Dorigoni:2024oft} that closely mirror the construction of single-valued MPLs in \cite{Frost:2023stm}. Similar unified formulae for motivic coactions at genus zero and genus one obtained from zeta generators are under investigation \cite{Kleinschmidt:2025dtk} and aim to complement the earlier literature on coaction formulae at genus one in physics \cite{Broedel:2018iwv, Wilhelm:2022wow, Forum:2022lpz} and mathematics \cite{Tapuskovic:2023xiu}.

The goal of the present work is to prove the conjectural formulae of \cite{Frost:2023stm} for the motivic coaction and the single-valued map of MPLs. First, the proof of the coaction formulae relies on the multivariate generalization \cite{Brown:2019jng} of the Ihara coaction formula \cite{Ihara1989TheGR} for generating series of MPLs subject to Knizhnik–Zamolodchikov equations. The action of the braid group on MPLs \cite{Kohno:CFT, KZB, Britto:2021prf} plays a key role in this proof, and we make use of concrete formulas for this action. The work of Ihara on Drinfeld associators, and especially the Ihara derivations \cite{Ihara1989TheGR, Ihara:stable, Furusho2000TheMZ} is also used at a key step, to recast the Ihara coaction formula in a form that better relates to our new formula. Second, the construction of single-valued MPLs via zeta generators is proven in two different ways: one proof is based on the relation between the motivic coaction and the single-valued map \cite{Brown:2013gia, DelDuca:2016lad} and a description of the antipode of MPLs \cite{Goncharov:2001iea} given in terms of zeta generators. An alternative proof proceeds by direct matching with the construction of single-valued MPLs in \cite{DelDuca:2016lad}. Throughout all of this work, we make repeated use of the identities of free Lie algebras and free associative algebras.

\subsection{Preview of main theorems and their motivation}

More specifically, the main results of this work are the proofs of Theorems \ref{thm:main} and \ref{thm:sv} below whose statements were conjectured in \cite{Frost:2023stm}. 

Let $\mathbb G_1$ denote a generating series of MPLs defined by iterated integrals of one-forms $\dd t /(t{-}a)$ on the configuration space of points on a punctured Riemann sphere, with alphabet $a \in \{ z_2,\ldots,z_n,0,1 \}$ and $t,z_1,\ldots,z_n \in \mathbb C$. The integral is over a homotopy class of paths from $0$ to $z_1$ avoiding the points in the alphabet. Note that the MPLs are multi-valued functions in the complex variables $z_1$, \ldots, $z_n$, and the generating series $\GG_{1}$ takes values in the degree completion of the universal enveloping algebra of the infinitesimal braid group as we will see later. The iterated-integral representation of MPLs (and MZVs as their special values) lets us view them as periods. 

Our first main Theorem \ref{thm:main} reformulates the motivic coaction $\Delta$ of the MPLs in $\mathbb G_1$ \cite{Goncharov:2001iea, Goncharov:2005sla, Brown:2011ik, BrownTate} in terms of the generating-series identity
\beq
\Delta \GG_1^\mot = \left( \HH_n^\dR \right)^{-1} \, \GG_1^\mot \, \HH_n^\dR \, \GG_1^\dR  \, , \ \ \ \ \ 
\HH^\dR_n = \MM^\dR \,\GG^\dR_n\, \cdots \GG^\dR_{2}\, ,
\label{prevone}
\eeq
where $\mathbb G_{j}$ ($j \geq 2$) is the generating series of MPLs associated with the path from $0$ to $z_j$ and one-forms ${\dd t}/(t{-}a)$ in the reduced alphabet $a \in \{ z_{j+1},\ldots,z_n,0,1 \}$. The MPLs in the generating series $\mathbb G_{j}$ ($1\leq j \leq n$) are accompanied by words in braid generators subject to well-studied bracket relations \cite{jimbo1989introduction, kassel2012quantum}.
The coaction of individual MPLs is obtained by taking coefficients of the series in (\ref{prevone}) once the right-hand side is cast into the same basis of words in the braid generators as the left-hand side (an accomplishment of this paper is to explain why and how this can be done). The superscripts $\mot$ and $\dR$ in (\ref{prevone}) refer to the lift of MPLs and MZVs to motivic and de Rham periods, respectively, and distinguish the first and second entry of the tensor product produced by the motivic coaction according to $X^\mot = X^\mot \otimes 1$ and $X^\dR = 1\otimes X^\dR$. Finally, $\MM^\dR$ in (\ref{prevone}) is a generating series for de Rham MZVs in their~$f$-alphabet representation (see Section \ref{sec:2.3} and \cite{Brown:2011ik, BrownTate}) accompanied by \emph{zeta generators} (denoted by $M_{2k+1}$ with $k\in\mathbb N$), which are derivations on the infinitesimal braid group. The properties of the zeta generators proposed in \cite{Frost:2023stm} will be derived in this work.

Our second main result, Theorem \ref{thm:sv}, expresses the result of applying the single-valued map $\sv$ to the generating series $\mathbb G_1$ in a form that is very similar to (\ref{prevone}),
\beq
\sv\, \GG_1 =  \left(\sv \,\HH_{n}\right)^{-1} \,\overline{\GG_1^t}\, \left(\sv \, \HH_{n}\right) \GG_1 \, , \ \ \ \ \ 
\sv\, \HH_n = (\sv\, \MM) \, (\sv \, \GG_n) \ldots (\sv\, \GG_2) \, ,
\label{prevtwo}
\eeq
where $\overline{\GG_1^t}$ denotes the complex conjugate of the series $\GG_1$ with a reversed concatenation order of its braid generators. Upon taking the coefficients of the series (\ref{prevtwo}) (in the sense of the comments below (\ref{prevone})), one obtains a reformulation of the single-valued map of individual MPLs constructed in \cite{svpolylog, Broedel:2016kls, DelDuca:2016lad}.

The conjugations by $\HH^\dR_n$, $\sv\, \HH_n$ in (\ref{prevone}), (\ref{prevtwo}) ensure that all zeta generators in $\MM^\dR$, $\sv\,\MM$ and braid generators in $\GG_{j}$ ($j\geq 2)$ enter the formulas as nested commutators acting on $\GG_{1}$. Both the applications and the proofs of (\ref{prevone}) and (\ref{prevtwo}) rely on the bracket relations between zeta generators and braid generators which ensure that their right-hand sides are expressible solely in terms of the braid generators of $\GG_{1}$. It is worth highlighting two practical advantages of these reformulations:
\begin{itemize}
\item The fibration bases (as in Section 8.5 of \cite{Duhr:2019tlz} or Example 4.3.1 in \cite{Panzer:2015ida}) of MPLs are preserved when expanding our formulas, namely, the alphabet of forms $\dd t/(t{-}a)$ entering the MPLs of $\GG_j$ is restricted to $a \in \{ z_{j+1},\ldots,z_n,0,1 \}$, i.e.\ excluding $z_1,\ldots,z_{j}$ (see Section \ref{sec:2.2.1}).
Similarly, the $f$-alphabet representation (discussed in Section \ref{sec:2.3})  of the de Rham and single-valued MZVs in the series $\MM^\dR$ and $\sv \, \MM$ automatically incorporate their relations over~$\mathbb Q$. 
\item The series $\MM^\dR$, $\sv \, \MM$ feature zeta generators $M_{2k+1}$ along with odd Riemann zeta values $\zeta_{2k+1}$ and compose products and higher-depth instances of MZVs with products $M_{2k_1+1} M_{2k_2+1}\ldots$ according to their coaction properties. On these grounds, the conjugations by $\MM^\dR$, $\sv \, \MM$ in (\ref{prevone}), (\ref{prevtwo}) manifest how terms in $\Delta$ and $\sv$ of MPLs that involve products of MZVs or indecomposable higher-depth MZVs (say $\zeta^\dR_{3,5}$ or $\sv \, \zeta_{3,3,5}$) can be computed systematically from the terms with depth-one~MZVs,~$\zeta_{2k+1}$.
\end{itemize}
These practical aspects complement the more conceptual motivation for this work, namely that the reformulation of the motivic coaction and the single-valued map in (\ref{prevone}) and (\ref{prevtwo}) make generalizations beyond genus zero more accessible.
This is expected because zeta generators can be described more universally as acting on the fundamental groups of punctured Riemann surfaces of genus zero and genus one by recovering a nodal Riemann sphere from  degeneration limits of tori (see e.g.\ Appendix A of \cite{Dorigoni:2024iyt}). 
In this broader setting, the genus-zero results of this work give a clear conceptual foundation for progress at higher genus. In practical terms, our approach delimits the possible changes of alphabet in generating series of iterated integrals beyond genus zero, which will be necessary for coaction formulae and single-valued versions.

As a first concrete evidence for this expectation, the one-variable instance $n=1$ of the single-valued map in (\ref{prevtwo}) exhibits striking parallels to the generating series of Brown's single-valued iterated Eisenstein integrals \cite{Brown:mmv, Brown:2017qwo, Brown:2017qwo2} in Section 3 of \cite{Dorigoni:2024oft}. 
Moreover, similar genus-one uplifts were proposed for the coaction formula (\ref{prevone}) in the one-variable case \cite{Kleinschmidt:2025dtk} (along with non-trivial consistency checks) and derived in the $n = 2$ case of (\ref{prevtwo})  \cite{Schlotterer:2025qjv}. All of these genus-one advances rely on the interplay of zeta generators with the fundamental groups of punctured spheres and tori. 
Going further, one can on the long run envision analogous formulas at genus $\geq 2$ once a suitable realization of zeta generators associated with higher-genus Riemann surfaces becomes available.

\subsection*{Acknowledgments}

We are grateful to Ruth Britto, Francis Brown, Steven Charlton, Eric D'Hoker, Daniele Dorigoni, Mehregan Doroudiani, Joshua Drewitt, James Drummond, Axel Kleinschmidt, Pierre Lochak, Carlos Mafra, Lionel Mason, Sebastian Mizera, Erik Panzer, Franziska Porkert, Leila Schneps, Oliver Schnetz, Yoann Sohnle and Matija Tapuskovic for combinations of inspiring discussions and collaboration on related topics. Moreover, we thank Axel Kleinschmidt and two anonymous referees for valuable comments on earlier versions of this work.

The research of HF is supported by the U.S. Department of Energy and the Sivian Fund, with additional support from the European Union (ERC, UNIVERSE PLUS, 101118787). Views and opinions expressed are however those of the author(s) only and do not necessarily reflect those of the European Union or the European Research Council Executive Agency. Neither the European Union nor the granting authority can be held responsible for them.
The research of MH is supported in part by the European Research Council under ERC-STG-804286 UNISCAMP and in part by the Knut and Alice Wallenberg Foundation under grant KAW 2018.0116. 
The research of DK was supported in part by the European Research Council under ERC-COG-724638 GALOP and by the Max Planck Institute for Mathematics in Bonn. 
The research of CR was supported by the European Research Council under ERC-STG-804286 UNISCAMP as well as ERC, UNIVERSE PLUS, 101118787.
The research of OS was supported by the European Research Council under ERC-STG-804286 UNISCAMP and the strength area ``Universe and mathematical physics'' which is funded by the Faculty of Science and Technology at Uppsala University.
The research of BV was supported by the Knut and Alice Wallenberg Foundation under grant KAW2018.0162.

\section{Multiple polylogarithms and the motivic coaction}
\label{sec:2}
Write $G(a_1,a_2,\ldots,a_w;z)$ for the weight $w$ multiple polylogarithm (MPL) defined by the iterated integral \cite{Goncharov:2001iea}
\beq
G(a_1,a_2,\ldots,a_w;z) = \int^z_0 \frac{ \dd t}{t-a_1} \, G(a_2,\ldots ,a_w;t)
\label{coact.01}
\eeq
for some labels $a_1,\ldots,a_w\in \CC$, and an argument $z \in \CC$, adopting the convention $G(\emptyset;z) = 1$. MPLs are multi-valued complex functions that depend on the
homotopy class of the integration path in (\ref{coact.01}) from $0$ to $z$ that avoids the $a_i$.
An important property of MPLs is that they satisfy the following shuffle-product identity:
\beq\label{eq:Gshuffle}
G(A;z)G(B;z)
= \sum_{C \in A \shuffle B}G(C;z)
\eeq
for ordered sets of labels $A,B,C$ --- or words --- with entries in some fixed set, or {\it alphabet}, ${\cal A} = \{z_1,\ldots,z_n\}$. A brief reminder of the shuffle product $C \in A \shuffle B$ can be found in Appendix \ref{app:Lie}.

Note that end-point divergences in the integral \eqref{coact.01} defining $G(a_1,\ldots,a_w;z)$ arise when $a_w=0$ or $a_1=z$. For this reason, we define regularized values at weight $w=1$
\beq
G(0;z) = \log(z), \qquad G(z;z) = - \log(z)\, .
\label{regw1}
\eeq
Divergent MPLs at higher weight $w\geq 2$ are then shuffle-regularized by imposing \eqref{eq:Gshuffle} which determines their regularized values from (\ref{regw1}) inductively in $w$ (see e.g.\ \cite{Panzer:2015ida}).

MPLs are closely related to multiple zeta values (MZVs). If we restrict to considering $a_i \in \{0,1\}$ in (\ref{coact.01}), then the limit of $G(a_1,\ldots,a_w;z)$ as $z\rightarrow 1$ yields the MZV
\begin{align}
\zeta_{n_1,n_2,\ldots,n_r} &= \sum^\infty_{0<k_1 <k_2<\ldots <k_r} 
k_1^{-n_1} k_2^{-n_2}\ldots k_r^{-n_r}
\label{coact.02}\\
&= (-1)^r \lim_{z\rightarrow 1} G(\underbrace{0,\ldots,0}_{n_r-1},1,\ldots,
\underbrace{0,\ldots,0}_{n_2-1},1,\underbrace{0,\ldots,0}_{n_1-1},1;z)
\notag
\end{align}
which is said to have \emph{depth} $r$ and \emph{weight} $n_1{+}\ldots{+}n_r$. The sums and
integrals converge if~$n_r \geq 2$, and we otherwise assign shuffle-regularized values through (\ref{eq:Gshuffle}) and 
\beq
G(0;1)= G(1;1)=0\, .
\label{w1reg}
\eeq
Both of (\ref{regw1}) and (\ref{w1reg}) are obtained by shifting the endpoints of the integration path in (\ref{coact.01}) by a small quantity $0<\epsilon \ll 1$ and defining regularized values as the zeroth-order term in the expansion of the convergent, $\epsilon$-dependent integral in $\log(\epsilon)$ as reviewed for instance in \cite{Panzer:2015ida} (see \cite{Deligne1989TheGR, Abreu:2022mfk} for a discussion of regularization in the context of tangential base points).

\subsection{Motivic and de Rham periods}
\label{sec:mdr}

The iterated-integral formulas \eqref{coact.01} define MPLs as multivalued complex functions. However, for some statements, it is important to treat the MPLs in a more formal sense as periods.\footnote{\label{ftnumberf}An MPL, which is a complex function, can be regarded as a period only when evaluated at an algebraic complex number and is otherwise referred to as a period function. When the labels $a_i$ and the endpoint of integration, $z$, take entries in algebraic numbers, they generate a number field, $k = \QQ(a_1, \ldots, a_w, z)$. The iterated-integral representation of MPLs then allows us to view them as periods of a motive defined over $k$.} In particular, we often make use of the motivic versions of MPLs, which are formal symbols that encode the data of the iterated integral and satisfy algebraic relations \cite{Francislecture, Kontsevich-Zagier}. We use a superscript $G^\mot$ to denote the motivic version of an MPL $G$. We write $\cal{P}^\mot$ for the algebra of motivic MPLs, MZVs and $(i \pi)^\mot$. This is a subalgebra of the algebra of motivic periods \cite{BrownTate, Francislecture}. MPLs can also be lifted to a \emph{de Rham} version, and we write $G^\dR$ for the de Rham version of an MPL $G$. Unlike the motivic periods, the de Rham periods are defined up to the discontinuities of the MPLs, which we can express schematically by writing $(i\pi)^\dR = 0$. We write ${\cal P}^\dR$ for the algebra of de Rham MPLs and MZVs, which is a subalgebra of the algebra of de Rham periods.\footnote{Formally, MPLs are periods of the pro-unipotent fundamental groupoid of the punctured Riemann sphere, which has an associated mixed Tate motive defined over the number field $k$ in footnote \ref{ftnumberf} of \cite{DG:2005}. The category of mixed Tate motives over the number field $k$ is Tannakian and is equivalent to the category of representations of an affine group scheme, referred to as the motivic Galois group and denoted by $\mathcal{G_{MT}}(k)$. In this context, ${\cal P}^\dR$ is the affine ring of functions of the motivic Galois group $\mathcal{G_{MT}}(k)$, restricted to de Rham versions of MPLs and MZVs, and is a connected, graded Hopf algebra over $\QQ$ \cite{Goncharov:2001iea, Goncharov:2005sla}.}

We recall that the algebra of de Rham periods is a Hopf algebra over $\QQ$, and the algebra of motivic periods is a Hopf comodule over it \cite{Francislecture}. In particular, ${\cal P}^\dR$ is a graded Hopf algebra, whose coproduct $\Delta$ and antipode $S$ were worked out concretely by Goncharov \cite{Goncharov:2005sla}. (See Section \ref{sec:7.1} for the definition of $S$.) This was extended to a motivic coaction of ${\cal P}^\dR$ on ${\cal P}^\mot$ by Brown \cite{BrownTate}. Most of this paper is devoted to studying this \emph{motivic coaction}. It is a map of $\QQ$-algebras
\beq
\Delta:\,\, {\cal P}^\mot \rightarrow {\cal P}^\mot \otimes_{\QQ} {\cal P}^\dR
\eeq
and will also be denoted by $\Delta$. See Section \ref{sec:2.4} for a review of the motivic coaction.

Setting $z=1$, a motivic or de Rham MPL, $G^\mot$ or $G^\dR$, defines a motivic or de Rham version of the corresponding MZV \cite{BrownTate}. We again use superscripts, $\zeta^\mot$ and $\zeta^\dR$, to denote these periods. The identities satisfied by motivic MZVs are similar to those satisfied by de Rham MZVs.  However, note that the even zeta values are zero as de Rham periods, and we write $\zeta_{2k}^\dR = 0$. For more details about MZVs and their $f$-alphabet description, see Section \ref{sec:2.3}.

Finally, we can recover numbers and complex functions from motivic periods by evaluating them using the period map, ${\rm per}$. We denote the result of applying ${\rm per}$ by removing the superscripts and writing $G$ or $\zeta$.

\subsection{Single-valued MPLs}
\label{sec:rwsv}

An MPL $G(a_1,\ldots,a_w;z)$ exhibits monodromies as the integration path from 0 to $z$ is deformed to wind around the singular points $a_1,\ldots,a_w$ of the integrand in (\ref{coact.01}). Given such an MPL, there exists a unique complex function, ${\rm sv} \, G(a_1,\ldots,a_w;z)$, expressed in terms of MPLs, their complex conjugates, and MZVs, such that it
\begin{itemize}
\item[(i)] is single-valued on the punctured Rieman sphere (i.e.\ it has trivial monodromy),
\item[(ii)] has vanishing (regularized) limit $z\rightarrow 0$, and 
\item[(iii)] satisfies the same holomorphic differential equations as $G(a_1,\ldots,a_w;z)$, namely:
\beq
\partial_z  {\rm sv} \, G(a_1,a_2,\ldots,a_w;z)
= \frac{ {\rm sv} \, G(a_2,\ldots,a_w;z) }{z-a_1} \, .
\label{coact.09} 
\eeq
\end{itemize}
These ${\rm sv} \, G(a_1,a_2,\ldots,a_w;z)$ are called \emph{single-valued MPLs}, and it is known from \cite{Brown:2013gia, Brown:2018omk} that there is a homomorphism, ${\rm sv}$, from motivic MPLs to single-valued MPLs, that sends the motivic MPL~$G^\mot(a_1,\ldots,a_w;z)$ to~${\rm sv}\,G^\mot(a_1,\ldots,a_w;z)$.

In the one-variable case, with $a_i \in \{0,1\}$, the explicit  construction of single-valued MPLs was given by Brown \cite{svpolylog}. This was generalised to single-valued MPLs in  
two variables (say $a_i \in \{0,1,y\}$ for some $y \in \mathbb C \setminus \{0,1\}$) in \cite{Broedel:2016kls} and to an arbitrary number of variables in 
\cite{DelDuca:2016lad}. See Section \ref{sec:7} for more details about this map ${\rm sv}$ and two proofs (Section \ref{sec:7.proof} and \ref{sec:7.3}) of our new formula in Theorem \ref{thm:sv} (previewed in (\ref{prevtwo})) encoding how the ${\rm sv}$ map acts on motivic MPLs.

\subsection{Generating series of MPLs}
\label{sec:2.2}

We will obtain our results by studying generating series of MPLs. In this section, we take care to define these generating series and introduce key notation.

First, introduce noncommuting formal variables $e_{a_i}$ for each $a_i$ in some alphabet, $\cal A$, of labels. For a word $W = a_1 \cdots a_w$ in ${\cal A}^\times$, we write $e_W = e_{a_1}\cdots e_{a_w}$ for the associated concatenation product of variables $e_a$. Then we define generating series, $\GG^\dR$ and $\GG^\mot$, of MPLs with labels in $\cal A$ by writing
\beq
 \GG\left[\begin{matrix} e_{a_1} & \cdots & e_{a_n} \\ a_1 & \cdots & a_n \end{matrix} ; z \right] = \sum_W e_{W^t}\,G(W; z) \, ,
\label{defgg}
\eeq
where the sum is over words $W$ in the alphabet $\cal A$ and $W^t$ is the reversed word. The sum includes the empty word, $W=\emptyset$, associated with $G(\emptyset;z)=1$. The MPLs $G(W; z)$ in (\ref{defgg}) associated with generic words $W=a_1\cdots a_w$ are viewed as functions of multiple complex variables, namely $z$ and all the non-constant $a_i \in {\cal A}$. We write $\GG[e_{a_1},\ldots,e_{a_n};z]$ for this generating series, as a convenient abbreviation. It follows from the definition of MPLs that the generating series (\ref{defgg}) solve Knizhnik--Zamolodchikov (KZ) equations
\beq
\partial_z \GG[e_{a_1},\ldots,e_{a_n};z] = \GG[e_{a_1},\ldots,e_{a_n};z] \sum_{i=1}^n \frac{e_{a_i}}{z-a_i} \, . \label{eq:dzG}
\eeq

The simplest non-trivial instance $\GG[e_0,e_1;z]$ of the generating series (\ref{defgg}) contains all MPLs at argument $z$ with labels in $\{0,1\}$. So, by \eqref{coact.02}, which gives MZVs as the special values of MPLs at $z=1$, it is natural to consider the limit
\beq
\lim_{z\rightarrow 1} \GG[e_0,e_1;z] = \Phi(e_0,e_1) \, .
\label{nvexpphi}
\eeq
This limit naively contains divergences, which we regularize by replacing MPLs that diverge at $z=1$ with their shuffle-regularised values (see the text below \eqref{coact.02}). The series $\Phi(e_0,e_1)$ is a (group-like, cf.\ Section \ref{sec:5.1}) generating series for shuffle-regularized MZVs \cite{LeMura}, and is called the \emph{Drinfeld associator} \cite{Drinfeld:1989st, Drinfeld2}. The Drinfeld associator takes values in the degree completion of the universal enveloping algebra of the free Lie algebra generated by $e_0,e_1$, and its inverse with respect to the concatenation of $e_i$ is specified by \cite{Drinfeld2} 
\beq\label{eq:PhiPhi1}
\Phi(e_0,e_1)\Phi(e_1,e_0) = \Phi(e_1,e_0)\Phi(e_0,e_1) = 1\, .
\eeq
The first few terms of the series expansion are given by
\beq\label{eq:Phiexpand}
\Phi(e_0,e_1) = 1 + [e_0,e_1] \zeta_2 + [[e_0,e_1],e_0{+}e_1]\zeta_3 + \ldots
\eeq
with words involving $\geq 4$ letters $e_i$ in the ellipsis.
In this series, each MZV in $\Phi(e_0,e_1)$ appears multiplied by certain polynomials in $e_0, e_1$. By the variety of relations among MZVs over $\mathbb Q$ \cite{Blumlein:2009cf, Jianqiang, GilFresan}, the form of these polynomials is \emph{not} unique. In view of this, we will later introduce a related but distinct generating series of MZVs, $\MM$, with a fixed conjectural $\mathbb Q$ basis, see section \ref{sec:2.3}.

\subsubsection{The KZ equation and fibration basis}
\label{sec:2.2.1}
In this paper, we study MPLs that depend on $n$ variables, $z_i$. It is convenient to initially define these MPLs as functions of $z_i$ on the real line, with a fixed ordering
\beq
z_0 < z_1 < z_2 < \cdots < z_n < z_{n+1}\, ,
\label{eq:zorder}
\eeq
where we set $z_0 = 0, z_{n+1} = 1$. MPLs depending on these variables can be analytically continued away from these real points. This will play a key role in Section \ref{sec:4.2}, where we relate different orderings via analytic continuation.

To define generating series for these MPLs, we begin with the KZ equation for a generating series ${\cal G}_n(z_1,\ldots,z_n)$ with values in noncommuting variables $e_{i,j} = e_{j,i}$:
\begin{equation}\label{eq:dkG}
\partial_k {\cal G}_n
= {\cal G}_n\, \Omega_k^{(n)}, \qquad \Omega_k^{(n)} = \sum_{\substack{i=0 \\ i\neq k}}^{n+1} \frac{e_{k,i}}{z_{ki}},
\end{equation}
where $z_{ij} = z_i {-} z_j$ and $\partial_k$ is differentiation by $z_k$. In fact, the commutativity of partial derivatives acting on ${\cal G}_n$ implies that the $e_{i,j}$ must satisfy the \emph{infinitesimal braid relations} \cite{jimbo1989introduction, kassel2012quantum},
\beq\label{eq:infbraid}
[e_{i,j},e_{k,\ell}] = 0,\qquad [e_{i,j}+e_{j,k},e_{i,k}]=0
\eeq
for distinct $i,j,k,\ell$. Geometrically, this appearance of the braid group is not merely a consequence of imposing commutativity of derivatives: the KZ connection lives on the relevant configuration space of marked points on the punctured sphere, whose monodromy is governed by the braid group, and the $e_{i,j}$ are the corresponding infinitesimal braid generators. The $e_{i,j}$ can then be thought of as generators for the braid group.

A key idea of this paper is to decompose the series ${\cal G}_n$ into a product of generating series which feature more restricted classes of MPLs, in the spirit of the \emph{fibration decomposition}; see Section 8.5 of \cite{Duhr:2019tlz} and Example 4.3.1 of \cite{Panzer:2015ida}. To this end, for each $k=1,2,\ldots,n$, define
\beq
\GG_k[\{e_{k,i}\};z_k] = \GG\left[\begin{matrix} e_{k,0}^\ast & e_{k,k+1} & \cdots & e_{k,n} & e_{k,n+1} \\ z_0 & z_{k+1} & \cdots & z_n & z_{n+1} \end{matrix} ; z_k \right], \qquad e_{k,0}^\ast = \sum_{i=0}^{k-1} e_{k,i}\, .
\label{adaptser}
\eeq
Note that we introduce here the linear combinations of braid generators $e^\ast_{k,0}$, which  satisfy
\begin{equation}\label{eq:easteast1}
[e_{i,0}^\ast, e_{j,0}^\ast] = 0, \qquad
[e_{i,0}^\ast, e_{j,k} ] = 0 
\end{equation}
for any $i< j, k$ or $i > j,k$. (Note that $[e_{i,0}^\ast, e_{j,k}] \neq 0$ for $j < i < k$.) The series $\GG_k$ contains MPLs with argument $z_k$ and depending only on the variables $z_k,z_{k+1},\ldots,z_n$ (recall that $z_0=0$ and $z_{n+1}=1$), namely MPLs of the form
\begin{equation}\label{eq:fibrationbasisMPL}
G(z_{i_1},z_{i_2},\ldots,z_{i_r};z_k), \qquad z_{i_j}\in \{0,1,z_{k+1},\ldots,z_n\}\, .
\end{equation}
The \emph{fibration basis} for MPLs in the variables $z_1,\ldots,z_n$ consists of finite products of MPLs in this form, for $k$ ranging over $k=1,\ldots,n$. Equivalently, an expression is in the fibration basis if every factor with argument $z_k$ has labels drawn only from $\{0,1,z_{k+1},\ldots,z_n\}$.

A general MPL depending on the $z_i$ does not belong to this fibration basis. For instance an MPL $G(\ldots,z_{k-1},\ldots;z_k)$ (for some $k\geq 2$) or $G(\ldots,z_{k},\ldots;z_k)$ (for some $k\geq 1$) is not of the form in (\ref{eq:fibrationbasisMPL}). However, such an MPL can be written in terms of products of MPLs that are in the fibration basis (possibly involving $\mathbb Q$-linear combinations of MZVs as coefficients) \cite{Brown:2009qja}. At the level of generating series, we have
\begin{equation}
{\cal G}_n(z_1,\ldots,z_n) = \GG_n \cdots \GG_2\GG_1 \, .
\label{calGdef}
\end{equation}
This equation captures how each MPL appearing in ${\cal G}_n$ on the left-hand side can be written on the right-hand side as sums of products (possibly with MZV coefficients) of MPLs in the fibration basis \eqref{eq:fibrationbasisMPL}. Equivalently, this factored product can be thought of as solving the KZ equation one variable at a time. Note that by \eqref{eq:dzG}, $\GG_k$ satisfies
\beq\label{eq:dkGk}
\partial_{k} \GG_k 
= \GG_k\, \bigg( \frac{e_{k,0}^\ast}{z_{k0}} + \sum_{\ell=k+1}^{n+1} \frac{e_{k,\ell}}{z_{k\ell}} \bigg)\, .
\eeq
So $\GG_1 = {\cal G}_1$ satisfies the $n=1$ instance of the KZ equation (\ref{eq:dkG}). At higher $n\geq 2$, however, the KZ equation (\ref{eq:dkG}) of the factorized expression (\ref{calGdef}) for ${\cal G}_n$ is less obvious and can be deduced from the flatness of the KZ connection $\sum_{k=1}^n \dd z_k \, \Omega_k^{(n)}$ (see for instance \cite{Britto:2021prf}): The general solution to (\ref{eq:dkG}) through a path-ordered exponential of $\sum_{k=1}^n \dd z'_k \, \Omega_k^{(n)}(z_1',\ldots,z_n')$ specializes to (\ref{calGdef}) when the endpoints of the integration path are taken to be $(z_1',\ldots,z_n')=(0,\ldots,0)$ and $(z_1',\ldots,z_n')=(z_1,\ldots,z_n)$. By flatness of the connection, the path-ordered exponential for these endpoints can be evaluated on the particular path composed of the straight lines between $(z_1',\ldots,z_n')=(0,\ldots,0,0) \rightarrow (0,\ldots,0,z_n) \rightarrow (0,\ldots,z_{n-1},z_n) \rightarrow \ldots \rightarrow (0,z_2,\ldots,z_n) \rightarrow (z_1,z_2,\ldots,z_n) $. Each of the path segments contributes one of the factors of $\GG_n \cdots \GG_2\GG_1$ to the path-ordered exponential between $(z_1',\ldots,z_n')=(0,\ldots,0)$ and $(z_1',\ldots,z_n')=(z_1,\ldots,z_n)$ which solves the KZ equation (\ref{eq:dkG}) by construction.
Note in particular that ${\cal G}_n(z_1,\ldots,z_n)$ satisfies the initial condition
\beq
\lim_{z_1\rightarrow 0}\lim_{z_2\rightarrow 0}\ldots \lim_{z_n\rightarrow0}{\cal G}_n(z_1,\ldots,z_n) = 1 \, ,
\eeq
where all limits are to be understood as shuffle-regularized MPLs. This characterizes \eqref{calGdef} to be the unique solution to the KZ equation with this initial condition.

The introduction of these generating series $\GG_k$ realizing the fibration basis decomposition plays an important role throughout the paper, for example in deriving the shuffle-regularized pinching limit of the series ${\cal G}_n$ in Lemma~\ref{lem:bigphiprod}.

Note that we often suppress the explicit dependence on the variables $z_i \neq 0,1$ in the notation for $\GG_k$. In particular, $\GG_1 = \GG_1[\{e_{1,i}\};z_1]$ depends on all $n$ variables, $z_1,\ldots, z_n$.

For each of the above generating series, we will use superscripts $\GG^\mot,\GG_k^\mot,{\cal G}_n^\mot$ or $\GG^\dR,\GG_k^\dR,{\cal G}_n^\dR$ when passing to the motivic and de Rham versions of the MPLs in the expansion (\ref{defgg}),~respectively.

\subsection{The motivic coaction for MPLs} 
\label{sec:2.4}

The motivic coaction on MPLs is traditionally obtained from the Goncharov--Brown formula \cite{Goncharov:2001iea, Goncharov:2005sla, Brown:2011ik, BrownTate}
\begin{align}
\Delta I^\mot(a_0;a_1,a_2,\ldots,a_w;a_{w+1}) &= \sum_{k=0}^w \sum_{0 = i_0 < i_1 < i_2 < \ldots < i_k < i_{k+1}= n+1}
I^\mot(a_0;a_{i_1}, a_{i_2},\ldots , a_{i_k};a_{n+1})
\notag \\
&\quad \times \prod_{p=0}^k I^\dR(a_{i_p};a_{i_p+1},\ldots ,a_{i_{p+1}-1};a_{i_{p+1}})\, ,
\label{GBcoact}
\end{align}
where the terms on the right-hand side are often visualized by inscribing polygons into a semi-circle.\footnote{The products of terms $I^\dR$ in the nested sums of (\ref{GBcoact}) can alternatively be compactly written using a modified integration contour that encircles the singular points of the integrand \cite{Goncharov:2001iea, DelDuca:2016lad, Abreu:2017mtm}.} The iterated integrals $I$ are recursively defined by
\beq
I(a_0;a_1,a_2,\ldots,a_w;a_{w+1})
= \int^{a_{w+1}}_{a_0} \frac{\dd t}{t-a_w} I(a_0;a_1,a_2,\ldots,a_{w-1};t)
\eeq
with $I(a_0;\emptyset;a_{w+1})=1$ and can always be reduced to the MPLs (\ref{coact.01}) at $a_0=0$  
\beq
I(0;a_1,a_2,\ldots,a_w;a_{w+1})
= G(a_w,\ldots,a_2,a_1;a_{w+1})
\label{iitogg}
\eeq
using the composition-of-paths formula for iterated integrals. The motivic and de Rham versions of (\ref{iitogg}) identify
the iterated integrals $I^\mot$ and $I^\dR$ on the right-hand side of (\ref{GBcoact}) as motivic and de Rham versions of MPLs. We take advantage of these superscripts to distinguish the first and second entry of the coaction to skip the $\otimes$ symbol of the notation $I^\mot= I^\mot\otimes 1$ and $I^\dR= 1\otimes I^\dR$ seen in many other references.

The coaction formula (\ref{GBcoact}) applies to arbitrary alphabets $a_i \in {\cal A}$, but the right-hand side does not preserve the fibration bases of the left-hand side (as defined below (\ref{eq:fibrationbasisMPL})). For instance, the coaction of MPLs $G^\mot(\ldots,a_i,\ldots;z_1)$ with $a_i \in \{0,1,z_2\}$ may yield terms of the form $G^\dR(\ldots,a_i,\ldots;1)$ which can be eventually expressed in the fibration basis of $G^\dR(\ldots,b_i,\ldots;z_2)$ with $b_i \in \{0,1\}$ (and $\mathbb Q$-linear combinations of MZV as coefficients) after some extra work.

\subsubsection{Multivariate Ihara formula}\label{sec:multivariate}
An alternative representation of the motivic coaction of MPLs in the mathematics literature is furnished by the Ihara coaction formula \cite{Ihara1989TheGR} and its multivariate generalization in Proposition 8.3 of \cite{Brown:2019jng}. The motivic associators ${\cal Z}^{k,\mot}$ and ${\cal Z}^{k,\dR}$ in \cite{Brown:2019jng}, with $k$ referring to the endpoint $z_k$ of the integration path, reduce to our generating series $\GG^\mot_k$ and $\GG^\dR_k$ (Section \ref{sec:2.2}) by suppressing the letters $e_2, \ldots, e_k$ in the expansion of ${\cal Z}^{k,\mot}$ and ${\cal Z}^{k,\dR}$.\footnote{\label{zfootnt}This is done to ignore the extra terms containing $\frac{\dd t}{t-z_{j}}$ for $j < k$ that appear in $\mathcal{Z}^{k}$.} Hence, Proposition 8.3 of \cite{Brown:2019jng} implies, in terms of our generating series,
\begin{equation}\label{eq:multi}
\Delta \GG^\mot_1[\{e_{1,i}\};z_1] = \GG^\mot_1[\{e'_{1,i}\};z_1] \GG^\dR_1[\{e_{1,i}\};z_1]\, ,
\end{equation}
where the letters $e'_{1,i}$ at $i\neq 0$ entering 
$\GG^\mot_1$ on the right-hand side are series by themselves
\begin{equation}
e'_{1,0} = e_{1,0} \, , \ \ \ \
e'_{1,k} = Z^\dR_k \,e_{1,k} \,(Z^\dR_k)^{-1}\, , \ \ \ \ k=2,3,\ldots,n{+}1\, .
\label{ykconj}
\end{equation}
The associators $Z^\dR_k$ on the right-hand side are defined as the shuffle-regularized limits\footnote{We intentionally use a different notation, $Z^\dR_{\ell}$, for the 
shuffle-regularized limits of $\GG_1^\dR$ as $z_1=z_\ell$ to indicate that these generating series incorporate the restriction of the ${\cal Z}^{k,\mot}$ in \cite{Brown:2019jng} as mentioned in footnote \ref{zfootnt}.}
\beq
Z^\dR_k(z_2,\ldots,z_n) = \lim_{z_1 \rightarrow z_k} \GG^\dR_1 \, .
\label{zkconj}
\eeq
Similar to the Goncharov--Brown formula (\ref{GBcoact}), the de Rham MPLs on the right-hand side of the multivariate Ihara formula (\ref{eq:multi}) which enter via the series $Z^\dR_k$ (equation (\ref{zkconj})) do not preserve the fibration bases of MPLs (defined around (\ref{calGdef})). When applying the Ihara formulas, the appearance of MPLs outside the fibration bases
introduces redundancies into the coaction formulae for specific MPLs inherited from the generating series.

Our main result, Theorem \ref{thm:main} below, is a formula for $\Delta \GG^\mot_1$ that removes this redundancy by preserving the fibration bases. Moreover the formula removes the redunancies that follow from the $\mathbb Q$ relations among the de Rham MZVs. This is achieved by forming a generating series for de Rham MZVs using the $f$-alphabet to be reviewed in Section \ref{sec:2.3}.

\subsubsection{Motivic coaction for MZVs}\label{sec:iharaMZV}
The form of the Goncharov--Brown formula (\ref{GBcoact}) does not depend on the size of the alphabet ${\cal A}$ used to define the MPLs on the left-hand side. By contrast, for the multivariate Ihara formula (\ref{eq:multi}), the equation is sensitive to $n$, the number of $z_i$ variables that appear in the MPLs in the series $\GG_1 = \GG_1[\{e_{1,i}\};z_1]$ (defined in (\ref{adaptser})). In the special case of $n=1$, the multivariate Ihara formula becomes simply\footnote{Setting $n=1$ and writing $e_0$ for $e_{0,1}$, $e_1$ for $e_{1,2}$, and $z$ for $z_1$, then $\GG^\mot_1[e_{0,1}^\ast,e_{1,2};z_1] = \GG^\mot[e_0,e_1;z]$.}
\beq\label{eq:1ihara}
\Delta \GG^\mot[e_0,e_1;z] = \GG^\mot[e_0,e_1';z] \GG^\dR[e_0,e_1;z] \, ,
\eeq
where now
\begin{equation}\label{eq:eprimephi}
e_1' = \Phi^\dR(e_0,e_1)\,e_1 \left(\Phi^\dR(e_0,e_1)\right)^{-1} \, .
\end{equation}
This formula can be used to obtain a formula for the coaction of motivic MZVs. Taking the (shuffle regularized) limit as $z\rightarrow 1$ of equation \eqref{eq:1ihara} gives
\begin{equation}\label{eq:iharaphi}
\Delta \Phi^\mot(e_0,e_1) = \Phi^\mot(e_0,e_1') \Phi^\dR(e_0,e_1)\, .
\end{equation}
In the same way as the coaction formula (\ref{GBcoact}) for MPLs is not the most practical way to extract simplified expressions for $\Delta G^\mot(a_1,\ldots,a_w;z)$, it may take significant work to simplify the expressions for $\Delta \zeta^\mot_{n_1,\ldots,n_k}$that follow from the coefficients of (\ref{eq:iharaphi}). This is because both $e_1'$ and the right-multiplicative series $\Phi^\dR(e_0,e_1)$ introduce de Rham MZVs, and it remains to mod out by their relations over $\mathbb Q$ when using the expansion of the Drinfeld associator that follows from (\ref{nvexpphi}). 

\subsection{$f$-alphabet and generating series $\MM$ of MZVs} 
\label{sec:2.3}

We recall the $f$-alphabet representation for motivic MZVs \cite{Brown:2011ik, BrownTate}. The $f$-alphabet is the set of noncommuting formal variables $\{f_3,f_5,f_7,\ldots\}$, indexed by the odd numbers larger than $1$. For notational convenience take ${\cal W} = \{3,5,7,\ldots\}^\times$ the set of all words in the odd numbers, $3,5,7,\ldots$, treated as symbols. For a word $W=ab\ldots c$ in ${\cal W}$, we write $f_W = f_af_b\ldots f_c$ for the corresponding term in the graded Hopf algebra $\mathbb{Q} \langle f_3, f_5, f_7, \ldots \rangle$ freely generated by the $f$-alphabet. The grading is given by assigning each $f_{2i+1}$ ($i\geq 1$) degree $2i+1$. The product on $\mathbb{Q} \langle f_3, f_5, f_7, \ldots \rangle$ is the shuffle product 
\beq
f_A \, \shuffle \, f_B  = f_{A\shuffle B}
\eeq
and the
coproduct is the deconcatenation coproduct given by
\begin{equation}\label{eq:fcoproduct}
\Delta \big( f_{W} \big) = \sum_{W=AB}   f_A  \otimes  f_B
\end{equation}
for words $A,B \in {\cal W}$. Consider also the graded commutative algebra $\mathbb{Q}[f_2] \otimes_{\mathbb{Q}} \mathbb{Q} \langle f_3, f_5, f_7, \ldots \rangle$ with the commuting symbol $f_2$ in degree 2. The product looks similar to before,
\beq
\big( f_2^k f_A\big) \shuffle \big( f_2^{\ell} f_B \big)  = f_2^{k+\ell} \, f_{A\shuffle B} \, .
\label{shuffprod}
\eeq
The algebra $\mathbb{Q}[f_2] \otimes_{\mathbb{Q}} \mathbb{Q} \langle f_3, f_5, f_7, \ldots \rangle$ is also a Hopf comodule over $\mathbb{Q} \langle f_3, f_5, f_7, \ldots \rangle$, with the trivial coaction $\Delta f_2 = f_2 \otimes 1$ on $f_2$. And in general, the coaction looks like
\begin{equation}\label{eq:fcoaction}
\Delta \big(f_2^k f_{W}\big) = f_2^k \sum_{W=AB}   f_A \otimes f_B
\end{equation}
for words $A,B,W \in {\cal W}$. Here the sum is over all words $A,B$ (including empty words) whose concatenation is $AB=W$. As we will see, the Hopf algebra on odd $f$-letters will model the $\mathbb{Q}$-algebra of de Rham MZVs and the Hopf comodule over it on all $f$-letters will model the $\mathbb{Q}$-algebra of motivic MZVs. 
Thus, going forward, we will introduce the $\mot$,  $\dR$ superscript on the $f$-letters to distinguish them and drop the tensor $\otimes$ notation in the coproduct or coaction formulae to be consistent with the rest of the article.

The $\mathbb{Q}$-algebra of de Rham MZVs (resp.\ motivic MZVs) can be mapped to the Hopf algebra (resp.\ Hopf-algebra comodule) on odd $f$-letters (resp.\ all $f$-letters) by a choice of some isomorphism $\phi$ subject to three defining properties \cite{Brown:2011ik, BrownTate}: 
\begin{itemize}
\item[(i)] The generators $f_a$ are normalized to be the images of Riemann zeta values, 
\beq
  \phi(\zeta^{\bullet}_a) = f_a^{\bullet} \, , \ \ \ \ a \geq 2\, ,
\label{normf}
\eeq
with $f_{2k}^{\mot} \coloneqq ({\zeta_{2k}} / {\zeta_2^k}) (f_2^{\mot})^k \in \mathbb Q (f_2^{\mot})^k$ in case of even weight $2k\geq 4$. Here and in the remainder of this section, we will employ the superscript $\bullet$ on the $f$-letters, MZVs and generating series thereof to indicate that the respective equations apply to both motivic versions ($\bullet \rightarrow \mot$)
and de Rham versions ($\bullet \rightarrow \dR$).
\item[(ii)] The isomorphism $\phi$ maps products of MZVs to shuffle products\footnote{This is automatic from the algebra morphism but is included here to emphasise that we do not use the concatenation product in the algebra on odd $f$-letters.} (\ref{shuffprod}) in the $f$-alphabet (not to be confused with the shuffle product (\ref{eq:Gshuffle}) of iterated integrals)
\begin{equation}
\label{eq:fshuffle}
\phi(\zeta^{\bullet}_{n_1,\ldots,n_r} \cdot \zeta^{\bullet}_{m_1,\ldots,m_s}) = \phi(\zeta^{\bullet}_{n_1,\ldots,n_r}) \shuffle \phi( \zeta^{\bullet}_{m_1,\ldots,m_s})
\end{equation}
for arbitrary $n_i,m_i \in \mathbb N$ with $n_r,m_s\geq 2$. 
\item[(iii)] The isomorphism $\phi$ maps the coproduct on de Rham MZVs (resp.\ coaction on motivic MZVs) in Section \ref{sec:2.4} to the deconcatenation coproduct (\ref{eq:fcoproduct}) (resp.\ coaction (\ref{eq:fcoaction}))
\beq\label{eq:fdelta}
\Delta\big( \phi(\zeta^{\bullet}_{n_1,\ldots,n_r}) \big) = \phi \big( \Delta(\zeta^{\bullet}_{n_1,\ldots,n_r}) \big)
\eeq
for arbitrary $n_i$ as above. 
\end{itemize}
In the absence of additional structure, the isomorphism $\phi$ is \emph{non-canonical}. This is because the images, $\phi(\zeta^\mot_{n_1,\ldots,n_r})$, of indecomposable\footnote{\label{indecfoot}Motivic MZVs $\zeta^
\mot_{n_1,\ldots,n_r}$ at depth $r\geq 2$ are referred to as {\it indecomposable} if they cannot be expressed in terms of $\mathbb Q$-linear combinations of products of $\zeta^\mot_w$ with $w\geq 2$ and indecomposable MZVs of lower weight. The analogous $\zeta_{n_1,\ldots,n_r}$ without the $\mot$ superscript are said to be indecomposable as well.} MZVs at higher depth ($r\geq 2$) can be shifted by rational multiples of $f_{n_1+\ldots +n_r}$, while preserving the defining properties (\ref{normf}), (\ref{eq:fshuffle}) and (\ref{eq:fdelta}). This gives a large amount of freedom in the definition of $\phi$, but see the end of this section for a possible choice.

For a given choice of $\phi$, we can now define a formal generating series of (motivic and de Rham) MZVs. Introduce a new alphabet of noncommuting variables, $M_3,M_5,M_7,\ldots$, known as \textit{zeta generators}. (Note that we are not considering any generators $M_a$ associated with even $a$ in this work.) These $M_{2k+1}$ will play an important role throughout this work and can be regarded as generators of the motivic Lie algebra of \cite{DG:2005, BrownTate}. We define a generating series $\MM$ as
\begin{align}
\MM^{\bullet} = \sum_{W} \phi^{-1}(f^{\bullet}_W)\, M_W \, ,
\label{defmser}
\end{align}
where $M_W = M_a M_b \ldots M_c$ for words $W=ab\ldots c$, and we sum over all words $W$ in ${\cal W}$ (including the empty word, with $M_\emptyset = 1$). Note that the commuting image $f_2$ of $\zeta_2$ does not occur in (\ref{defmser}), and we will mostly use the de Rham version $\MM^\dR = 1 + \sum_{k=1}^
\infty \zeta^\dR_{2k+1} M_{2k+1}+\ldots$ in our main results below with words $f_a f_b\ldots$ comprising $\geq 2$ odd subscripts $a,b$ in the ellipsis. As will be detailed in Proposition \ref{prop:1thm} and Remark \ref{imprmk}, the choice of $\phi$ in the definition (\ref{defmser}) of $\mathbb M^\bullet$ will not affect our main results since a change of the isomorphism will induce a compensating redefinition of the zeta generators, also see Remark \ref{rmk22}.

The coaction (\ref{shuffprod}) of $f_2^k f_W$ and the property \eqref{eq:fdelta} of the $\phi$ isomorphism imply~that
\beq \label{eq:Mcoaction}
\Delta \MM^\mot = \MM^\mot  \MM^\dR
\eeq
which does not feature any analogue of the change of alphabet in the coaction formula for the Drinfeld associator, see \eqref{eq:iharaphi}.

\begin{remark}
It is possible to make a canonical choice for the isomorphism $\phi$ by introducing some additional constraints. In \cite{Keilthy, Dorigoni:2024iyt}, an inner product on words in braid generators is used to induce a canonical choice of $\phi$. For example, up to weight $11$, $\phi$ is parametrized by
\begin{align}
\phi( \zeta^\mot_{3,5}) &= - 5 f_3 f_5 + q_8 f_8 \, , \label{appMZV.g6} \\
\phi( \zeta^\mot_{3,7}) &= - 14 f_3 f_7 - 6 f_5 f_5 + q_{10} f_{10}\, , \notag \\
\phi( \zeta^\mot_{3,3,5}) &= - 5 f_3 f_3 f_5 - 45 f_2 f_9 - \frac{6}{5} f_2^2 f_7  + \frac{4}{7} f_2^3 f_5 +q_{11} f_{11} \, 
\notag
\end{align} 
for \emph{some} free rational coefficients, $q_8,q_{10},q_{11}$. The method in \cite{Dorigoni:2024iyt} fixes these coefficients to be $q_8=\frac{100471}{35568}$, $q_{10}=\frac{408872741707}{40214998720}$ and $q_{11}= \frac{1119631493}{14735232}$. Numerous earlier references (including \cite{Brown:2011ik, Schlotterer:2012ny}) chose instead $q_8=q_{10}=q_{11}=0$ and at higher weight take the coefficients of $f_{n_1+\ldots +n_r}$ in $\phi(\zeta^\mot_{n_1,\ldots,n_r})$ to be zero, for all indecomposable MZVs (see footnote \ref{indecfoot}) in a conjectural reference basis (for instance the one in the MZV datamine \cite{Blumlein:2009cf}).
\end{remark}

\begin{remark}
\label{rmk22}
We will mostly consider de Rham MZVs in this paper. For de Rham MZVs, the choice of even weight coefficients, like $q_8$ and $q_{10}$ in (\ref{appMZV.g6}), drops out by $\zeta^\dR_{2k}=0$. Indeed, upon inverting (\ref{appMZV.g6}) 
\begin{align}
\phi^{-1}(f_3 f_5) &= - \frac{1}{5} \zeta^\mot_{3, 5} +\frac{q_8}{5}\zeta^\mot_8 \, , \label{phiinv}\\
\phi^{-1}(f_3 f_7) &=  - \frac{1}{14} \zeta^\mot_{3,7}   - \frac{ 3}{14} (\zeta^\mot_5)^2 
+\frac{q_{10}}{14}\zeta^\mot_{10} \, ,
\notag \\
\phi^{-1}(f_3 f_3 f_5) &= -\frac{1}{5} \zeta^\mot_{3, 3, 5} +
\frac{1}{2}\zeta^\mot_5 \zeta^\mot_6 - 
  \frac{3}{5}  \zeta^\mot_7 \zeta^
  \mot_4 -
  9 \zeta^\mot_9 \zeta^\mot_2+
  \frac{q_{11}}{5} \zeta^\mot_{11}  \, ,\notag
\end{align}
and passing to de Rham periods, $\zeta_8^\dR$ and $\zeta_{10}^\dR$ are $0$, and lead to $\phi^{-1}(f^\dR_3 f^\dR_5) = - \frac{1}{5} \zeta^\dR_{3, 5}$ and $\phi^{-1}(f^\dR_3 f^\dR_7) =  - \frac{1}{14} \zeta^\dR_{3,7}   - \frac{ 3}{14} (\zeta^\dR_5)^2 $.

However, the choice of odd-weight coefficients does still matter for de Rham MZVs. For instance, the coefficient $q_{11}$ in (\ref{appMZV.g6}) and (\ref{phiinv}) first appears in $\MM^\dR$ in the coefficient of $\zeta_{11}^\dR$: $\frac{1}{5} q_{11}[M_3, [M_3,M_5]] \zeta^\dR_{11}$. Still, different choices of $\phi$ can be absorbed into redefinitions of the zeta generator $M_{a}$ at odd $a\geq 11$ by (rational multiples of) nested brackets of $M_b$ with $b< a{-}5$, for instance redefinitions of $M_{11}$ by $[M_3, [M_3,M_5]]$. This is an example of the well-known freedom to redefine zeta generators by nested brackets of their lower-weight counterparts unless an inner-product structure is used to identify canonical choices of zeta generators  \cite{Keilthy, Dorigoni:2024iyt}. We will see in Proposition \ref{prop:1thm} below that our results apply to \emph{any} choice of the isomorphism $\phi$ along with an adapted action of the zeta generators. In particular, these adapted actions implement the redefinition of $M_{11}$ by rational multiples of $[M_3, [M_3,M_5]]$ if the choice of $q_{11}$ in (\ref{appMZV.g6}) is modified.
\end{remark}

\section{Overview of results}
\label{sec:3}

It is the purpose of this paper to prove a new formula for the motivic coaction on MPLs which has computational advantages and facilitates generalizations beyond genus zero. In Section \ref{sec:6.2}, we will define a family of adjoint-like actions of the zeta generators $M_k$ (for $k=3,5,7,\ldots$) on the braid generators $e_{i,j}$ that takes the form
\beq
[e_{1,0}, M_k] = 0
\, , \ \ \ \
[e_{1,\ell}, M_k] = \sum_{r=1}^{2\ell - 3}[P_k^{(r)}, e_{1,\ell}] \, , \ \ \ \ \ell = 2,3,\ldots,n{+}1
\eeq
for certain Lie polynomials $P^{(r)}_k$ of degree $k$ in the braid generators.\footnote{In the notation of \cite{Frost:2023stm}, $P^{(r)}_k = W_k(E_0^{(r)}, E_1^{(r)})$ for the combinations $E_i^{(r)}$ of braid generators in (\ref{eq:E0E1}) below. 
The Lie polynomials $g_k$ of \cite{Dorigoni:2024iyt} with odd $k$ are obtained from the $W_k$ in this work by replacing $(e_0,e_1) \rightarrow (x,-y)$.} With this understood, our result is\footnote{Our conventions here differ from those of \cite{Frost:2023stm}, which fixes the opposite ordering, $z_{n+1}=0<z_n<\cdots < z_1 < z_0 = 1$, from the one used here. With that ordering, $\GG_k$ here becomes $\GG_{n-k+1}$ in \cite{Frost:2023stm}, and, in Theorem \ref{thm:main}, the series $\HH_{n+1}$ becomes the series $\FF_n = \MM \GG_1 \GG_2 \ldots \GG_n$ in \cite{Frost:2023stm}.}

\begin{theorem}\label{thm:main}
The motivic coaction acts on the generating series $\GG_1$ as
\begin{equation}
\Delta \GG_1^\mot = \left( \HH_n^\dR \right)^{-1} \, \GG_1^\mot \, \HH_n^\dR \, \GG_1^\dR \, ,\qquad \text{with}\qquad \HH^\dR_{n} = \MM^\dR \,\GG^\dR_n\, \cdots \GG^\dR_{2} \, .
\label{thm2.1eq}
\end{equation}
\end{theorem}
This formula (\ref{thm2.1eq}) for the action of $\Delta$ on our generating series (\ref{adaptser}), first conjectured in \cite{Frost:2023stm}, makes it easier to extract simplified formulas for the coaction of specific MPLs by organising the computation as conjugation by the series $\HH_n^\dR$. Note that, for $n=1$, the conjugating series specialises to $\HH_1^\dR = \MM^\dR$ defined in (\ref{defmser}), and Theorem \ref{thm:main} reduces to a formula for the coaction of the generating series $\GG_1 \, |_{n=1} = \GG[e_0,e_1;z]$ in (\ref{eq:dzG}) of MPLs in one variable $z=z_1$:
\beq
\Delta  \GG^\mot[e_0,e_1;z]
=  (\MM^\dR)^{-1} \, \GG^\mot[e_0,e_1;z] \, \MM^\dR \,  \GG^\dR[e_0,e_1;z]
\label{thm1var}
\eeq
One may view the conjugation by $\HH_n^\dR$, $\MM^\dR$ in (\ref{thm2.1eq}), (\ref{thm1var}) as interpolating between the deconcatenation coproduct $a_1 a_2 \cdots a_r \mapsto \sum_{j=0}^r a_1  \cdots a_j\otimes a_{j+1}  \cdots a_r$ on words in the alphabet $a_i \in \{ z_2,\ldots,z_n,0,1\}$ of $\mathbb G_1$ and the motivic coaction of MPLs. The non-trivial contributions of $\HH_n^\dR = 1+\ldots$ to the expansion of (\ref{thm2.1eq}) determine the corrections involving de Rham MPLs in fewer variables and de Rham MZVs to the terms
\beq
\Delta G^\mot(a_1,a_2,\ldots,a_r;z_1) = \sum_{j=0}^r G^\mot(a_{j+1},\ldots,a_r;z_1) \, G^\dR(a_1,\ldots,a_j;z_1) +\ldots
\label{naivedel}
\eeq
which correspond to a naive deconcatenation coaction.

As an application of Theorem \ref{thm:main}, we also give a new formula for the generating function of \emph{single-valued} MPLs. In Section \ref{sec:7}, we show that Theorem \ref{thm:main} implies our second main theorem:

\begin{theorem}\label{thm:sv}
The single-valued map acts on the generating series $\GG_1$ as
\begin{equation}
\sv \, \GG_1^\mot = \left(\sv \,\HH_{n}^\mot\right)^{-1} \,\overline{\GG_1^t}\, \left(\sv \, \HH_{n}^\mot\right) \GG_1 \, ,
\label{thm2.2eq}
\end{equation}
where $\overline{\GG_1^t}$ is obtained from the complex conjugate of $\GG_1$ by reversing the concatenation order in the braid generators. The series $\HH_n$ is as in Theorem \ref{thm:main}, and
\beq
\sv\, \HH_n^\mot = (\sv\, \MM^\mot) \, (\sv \, \GG_n^\mot) \ldots (\sv\, \GG_2^\mot) \, . 
\label{defsvH}
\eeq
Moreover,\footnote{As discussed in Section \ref{sec:rwsv}, $\sv$ is a map from motivic MPLs to single-valued MPLs, which we apply term-wise in the series expansion of $\HH_n$.}
\beq
\sv\, \HH_n^\mot = \overline{\HH_n^t} \, \HH_n \, .
\eeq
\end{theorem}

The formula (\ref{thm2.2eq}) in Theorem \ref{thm:sv} appeared as a conjecture in \cite{Frost:2023stm}. Note that, for $n=1$ with $\sv\, \HH_1^\mot=\sv\, \MM^\mot = \mathbb M^t \mathbb M$, Theorem \ref{thm:sv} specializes to the generating series
\beq
\sv \, \GG^\mot[e_0,e_1;z]
=  (\sv \, \MM^\mot)^{-1} \, \overline{ \GG^t[e_0,e_1;z] } \, (\sv \, \MM^\mot) \,  \GG[e_0,e_1;z]
\label{svin1vr}
\eeq
of single-valued MPLs in one variable $z_1=z$.

Similarly to the comment below (\ref{thm1var}), the conjugation by $\sv\, \HH^\mot_n$, $\sv \mathbb M^\mot$, in (\ref{thm2.2eq}), (\ref{svin1vr}) connects the single-valued map of MPLs with the simple operation $a_1 a_2 \ldots a_r  \rightarrow \sum_{j=0}^r \overline{a_j\ldots a_1} \shuffle a_{j+1}\ldots a_r$ on words in the alphabet $a_i \in \{ z_2,\ldots,z_n,0,1\}$ of $\mathbb G_1$. The non-trivial contributions of $\sv\, \HH^\mot_n$ to the expansion of (\ref{thm2.2eq}) determine the corrections involving single-valued MZVs and MPLs in fewer variables to the particularly simple terms
\beq
\sv\, G^\mot(a_1,a_2,\ldots,a_r;z_1) = \sum_{j=0}^r \overline{G(a_{r},\ldots,a_{j+2},a_{j+1};z_1)} \, G(a_1,a_2,\ldots,a_j;z_1) +\ldots\, .
\label{nvsvmp}
\eeq
The rest of this section gives a structured overview of how these two main theorems are proved.

\subsection{Connection to the multi-variate Ihara formula}
\label{sec:3.1}
The key idea towards proving the new coaction formula, Theorem \ref{thm:main}, is to start from the multivariate Ihara formula of \cite{Brown:2019jng} (reviewed in Section \ref{sec:multivariate}). Indeed, equation \eqref{thm2.1eq} of Theorem \ref{thm:main} can be rewritten to closely resemble to the Ihara formula. Consider expanding the generating series $\GG_1^\mot$ in the braid generators ($e_{1,0}^{\ast}=e_{1,0}$ or $e_{1,\ell}$, with $\ell = 2,\ldots,n{+}1$). By inserting $1= \HH^\dR_{n}(\HH^\dR_{n})^{-1}$ between any pair of braid generators, \eqref{thm2.1eq}~of the theorem becomes
\beq
\Delta \GG^\mot_1[e_{1,0} ,\{e_{1,\ell}\};z_1] =   \GG_1^\mot\big[(\HH^\dR_n)^{-1}e_{1,0} \HH^\dR_{n} \, ,\,\{(\HH^\dR_n)^{-1}e_{1,\ell} \HH^\dR_{n}  \} ; z_1 \big] \, \GG_1^\dR[e_{1,0} ,\{e_{1,\ell}\};z_1]\, .
\label{outprf.03}
\eeq
More explicitly, given a term $e_{1,i_1}e_{1,i_2}\cdots e_{1,i_r}\, G^\mot(z_{i_r},\ldots,z_{i_2},z_{i_1};z_1)$ in the expansion (\ref{defgg}) of $\GG_1^\mot$, then its contribution to $(\HH^\dR_{n})^{-1}\GG_1^\mot\HH^\dR_{n}$ becomes
\beq
(\HH^\dR_{n})^{-1}e_{1,i_1}\HH^\dR_{n}(\HH^\dR_{n})^{-1}e_{1,i_2}\HH^\dR_{n}(\HH^\dR_{n})^{-1}\cdots \HH^\dR_{n}(\HH^\dR_{n})^{-1}e_{1,i_r}\HH^\dR_{n}\, G^\mot(z_{i_r},\ldots,z_{i_2},z_{i_1};z_1)\, .
\eeq
In terms of the generating series, the conjugation by the series $\HH^\dR_n$ then becomes a \emph{change of alphabet} in the braid generators, and we write
\beq
\tilde{e}_{1,\ell} = (\HH^\dR_n)^{-1} e_{1,\ell} \HH^\dR_{n} \, .
\label{prekey}
\eeq
The letter $e_{1,0}$ is left unchanged, $\tilde{e}_{1,0} = e_{1,0}$, as it commutes both with the series $\GG_{2},\ldots,\GG_{n}$ as well as with the series $\MM$.\footnote{See \eqref{eq:adaction} in Section \ref{sec:6.2}, where $e_{1,0} = E_0^{(1)}$ in the notation of that section.} So \eqref{outprf.03} may be written as
\beq
\Delta \GG^\mot_1[e_{1,0},\{e_{1,\ell}\} ;z_1] =   \GG_1^\mot [ e_{1,0}   , \{\tilde{e}_{1,\ell}\} ; z_1 ] \, \GG_1^\dR[e_{1,0} ,\{e_{1,\ell}\}; z_1]\, .
\label{outprf.04}
\eeq
This version of the formula is not as suitable for practical calculations as \eqref{thm2.1eq} (Theorem \ref{thm:main}), since we now have to carefully expand each of the series $\tilde{e}_{1,\ell}$ while also expanding $\GG_1^\mot$. However, we will use this form, \eqref{outprf.04}, to prove the Theorem.

Note that \eqref{outprf.04} resembles the multivariate Ihara formula, equation \eqref{eq:multi} (Section \ref{sec:multivariate}). To prove our Theorem from the Ihara formula, we need to show that the series $e_{1,\ell}' = Z^\dR_{\ell}
e_{1,\ell} (Z^\dR_{\ell})^{-1}$ that appear in the Ihara formula (see \eqref{ykconj}) are equal to the series $\tilde{e}_{1,\ell} = (\HH^\dR_n)^{-1} e_{1,\ell} \HH^\dR_{n}$ that appear in \eqref{outprf.04}. Equivalently, we prove that
\beq
(\MM^\dR)^{-1} e_{1,\ell} \MM^\dR = (\GG^\dR_n\cdots\GG^\dR_2 Z^\dR_\ell)\, e_{1,\ell}\, (\GG^\dR_n\cdots\GG^\dR_2 Z^\dR_\ell)^{-1} 
\label{keyeq}
\eeq
for all of $\ell=2,\ldots,n{+}1$.
This central identity is Lemma \ref{lem:DIcon} in Section \ref{sec:6.2}, and is the key bridge between the Ihara formula and Theorem \ref{thm:main}.

\subsection{Proof strategy for the coaction formula}
\label{sec:3.3}

There are three main steps in our proof of equation (\ref{keyeq}) (i.e.\ Lemma \ref{lem:DIcon} which is the principal work required to prove Theorem \ref{thm:main}):
\begin{itemize}
\item[(i)] The left-hand side of (\ref{keyeq}) is manifestly constant with respect to the variables $z_i$. Whereas the individual series $\GG^\dR_k$ and $Z^\dR_\ell$ appearing on the right-hand side do depend on $z_2,\ldots,z_n$. Our first step, Lemma \ref{lem:YeY}, is to prove that the right-hand side of (\ref{keyeq}) is in fact independent of the variables $z_i$. 
\item[(ii)] Not only is the right-hand side of (\ref{keyeq}) independent of $z_i$, we also show that it can be written in terms of Drinfeld associators. In Lemma \ref{lem:bigphiprod} we show that the right-hand side of \eqref{keyeq} may be written as
\begin{align}
(\GG^\dR_n\cdots\GG^\dR_2 Z^\dR_\ell)\, e_{1,\ell}\, (\GG^\dR_n\cdots\GG^\dR_2 Z^\dR_\ell)^{-1}
= \left(\Phi^{(1)}  \ldots \Phi^{(2\ell-3)}\right) e_{1,\ell} \left(\Phi^{(1)} \ldots \Phi^{(2\ell-3)}\right)^{-1} \, ,\label{simpkey}
\end{align}
where the Drinfeld associators, $\Phi^{(r)} = \Phi^\dR(E_0^{(r)}, E_1^{(r)})$, are given in terms of certain sums, $E^{(r)}_0$ and $E^{(r)}_1$, of the braid generators $e_{i,j}$, specified in \eqref{eq:E0E1} below. These particular combinations arise due to the braid action on MPLs, which is reviewed in Section \ref{sec:4.2}.
\item[(iii)] Finally, we show in Lemma \ref{lem:DIcon} that \eqref{simpkey} does indeed equal $(\MM^\dR)^{-1} e_{1,\ell} \MM^\dR$. This requires a number of results about Drinfeld associators, that we prove in Sections \ref{sec:5.1} and \ref{sec:6.1}. The proof of Theorem \ref{thm:main} then follows in Section \ref{sec:6.2}.
\end{itemize}

\subsection{Proof strategy for the single-valued map formula}
\label{sec:3.4}
The new formula for the single-valued map, Theorem \ref{thm:sv}, is proved in Section \ref{sec:7}. To better connect our result to earlier literature, we present two equivalent proofs of this formula:
\begin{itemize}
\item[(i)] 
The first proof is based on a purely combinatorial description of the single-valued map which is determined by the motivic coaction \cite{Brown:2013gia, DelDuca:2016lad}. Section \ref{sec:7.1} studies how the antipode acts on our generating series of MPLs and MZVs, and the link to the single-valued map is presented in Section \ref{sec:7.2}. Then Section \ref{sec:7.proof} proves our single-valued formula as a corollary of Theorem~\ref{thm:main}.
\item[(ii)] The second proof is based on the formulas for single-valued MPLs in any number of variables given in \cite{DelDuca:2016lad} (see also \cite{svpolylog} and \cite{Broedel:2016kls} for the similar formulas in one and two variables). The formulas in \cite{DelDuca:2016lad} use generating series identical to ours and a change of alphabet akin to that in the multivariate Ihara formula (equation (\ref{eq:multi})). See Section \ref{sec:7.3} for further details.
\end{itemize}

\section{Identities for generating series of MPLs}
\label{sec:4}

In this section, we carry out the first two steps (i) and (ii) of Section \ref{sec:3.3} towards proving Theorem \ref{thm:main}. Step (i) is Lemma \ref{lem:YeY} and step (ii) is Lemma \ref{lem:bigphiprod}.

\subsection{Differential equations} 
\label{sec:4.1}
In Section \ref{sec:2.2}, we defined the generating series ${\cal G}_n(z_1,\ldots,z_n) = \GG_n \cdots \GG_1$ of MPLs in $n$ variables which satisfies the KZ equation, $\partial_k {\cal G}_n = {\cal G}_n \Omega_k^{(n)}$ (see equation \eqref{eq:dkG}). For $k\neq \ell$, we have the commutativity of partial derivatives, $[\partial_k, \partial_\ell] {\cal G}_n = 0$, as a consequence of the \emph{infinitesimal braid relations}, equation \eqref{eq:infbraid}. In this section, we study the generating series defined by ($\ell = 2,\ldots,n{+}1$)
\beq
\YY_\ell = \GG_n\ldots \GG_2 Z_\ell = \lim_{z_1 \rightarrow z_\ell} {\cal G}_n \, ,
\label{def:yell}
\eeq
with the shuffle-regularized limit $Z_\ell = \lim_{z_1 \rightarrow z_\ell} \GG_1$. As outlined in Section \ref{sec:3}, the differential equations satisfied by $\YY_\ell$ play an important role in our proofs.

\begin{lemma}\label{lem:dkY}
For $k\neq \ell$,
\beq
\partial_k\YY_\ell = \YY_\ell \bigg(\frac{e_{k,1}}{z_{k\ell}} + \sum^{n+1}_{\substack{i= 0\\ i\neq 1, k}}\frac{e_{k,i}}{z_{ki}}\bigg)\, ,\qquad
\partial_\ell \YY_\ell = \YY_\ell  \sum_{\substack{i=0\\i\neq \ell}} \frac{e_{1,i}+e_{\ell,i}}{z_{\ell i}}  \, .
\eeq
\end{lemma}

\begin{proof}
For $k \neq \ell$, we take the $z_1 \rightarrow z_\ell$ limit of the KZ equation \eqref{eq:dkG} to find:
\beq
\partial_k \YY_\ell = \lim_{z_1 \rightarrow z_\ell} \big( {\cal G}_n \Omega_k^{(n)} \big) = \YY_\ell \lim_{z_1 \rightarrow z_\ell} \Omega_k^{(n)}\, ,
\eeq
which gives the first part of the Lemma. For the derivative $\partial_\ell \YY_\ell$, we have to be more careful. It is helpful to write $\YY_\ell$ as a (shuffle regularized) integral:
\beq
\YY_\ell = {\cal G}_n(z_1,z_2,\ldots,z_n) + \int_{z_1}^{z_\ell} \dd z \, \partial_z {\cal G}_n(z,z_2,\ldots,z_n) \, .
\eeq
Then,
\beq
\partial_\ell \YY_\ell = \partial_\ell \, {\cal G}_n(z_1,z_2,\ldots,z_n) + \lim_{z_1\rightarrow z_\ell} \partial_1 \, {\cal G}_n(z_1,z_2,\ldots,z_n) + \int_{z_1}^{z_\ell} \dd z \, \partial_z \partial_\ell \, {\cal G}_n(z,z_2,\ldots,z_n)\, .
\eeq
Using the KZ equation, \eqref{eq:dkG}, and integrating over the total derivative
\beq
\partial_\ell \YY_\ell = \YY_\ell \lim_{z_1 \rightarrow z_\ell} \left(\Omega_1^{(n)} + \Omega_\ell^{(n)} \right),
\label{elln1}
\eeq
which gives the second part of the Lemma.
\end{proof}

\begin{lemma}\label{lem:YeY}
For $2\leq k \leq n$ and $2\leq \ell \leq n{+}1$, i.e.\ including the case of $k=\ell$, we have
\begin{equation}
\partial_k \left( \YY_\ell\, e_{1,\ell} \, (\YY_\ell)^{-1} \right) = 0 \, .
\end{equation}
\end{lemma}

\begin{proof}
Consider first when $k\neq \ell$. Then, by Lemma \ref{lem:dkY},
\beq
\partial_k \left( \YY_\ell\, e_{1,\ell} \, (\YY_\ell)^{-1} \right) = \YY_\ell \left[\frac{e_{k,1}}{z_{k\ell}} + \sum^{n+1}_{\substack{i= 0\\ i\neq 1, k}}\frac{e_{k,i}}{z_{ki}} , e_{1,\ell} \right] (\YY_\ell)^{-1} \, .
\label{rlem42.1}
\eeq
But this vanishes by the infinitesimal braid relations \eqref{eq:infbraid}, more specifically by $[e_{k,i},e_{1,\ell}]=0$ along with $z_{ki}^{-1}$ with $i\neq \ell$ and by $[e_{k,1}{+}e_{k,\ell},e_{1,\ell}]=0$ along with $z_{k\ell}^{-1}$.

In the case that $k=\ell$, Lemma \ref{lem:dkY} gives
\beq
\partial_\ell \left( \YY_\ell\, e_{1,\ell} \, (\YY_\ell)^{-1} \right) = \YY_\ell \left[ \sum^{n+1}_{\substack{i=0\\i\neq 1,\ell}} \frac{e_{1,i}+e_{\ell,i}}{z_{\ell i}}, e_{1,\ell} \right] (\YY_\ell)^{-1}\, ,
\label{rlem42.2}
\eeq
where the $i=1$ term of the sum is absent since the pole $z_{\ell 1}^{-1}$ does not occur in the limit $z_1 \rightarrow z_\ell$ of (\ref{elln1}).
But we again have the braid relation $[e_{1,i} + e_{\ell,i}, e_{1,\ell}] = 0$, for each $i \neq 1,\ell$, so the right-hand side of (\ref{rlem42.2}) vanishes just like that of (\ref{rlem42.1}) which implies the Lemma.
\end{proof}

\subsection{The braid group action} 
\label{sec:4.2}
For $n\geq 2$, write $B_{n}$ for the braid group on $n$ strands, which is generated by the simple braids $\sigma_{i} = \sigma_{i, i+1}$ (with $1 \leq i \leq n{-}1$) modulo the relations \cite{Kohno:CFT}
\begin{align}
\sigma_{i} \, \sigma_{j} &= \sigma_{j} \, \sigma_{i}  \, , &&\text{for} \; \; \vert i {-} j \vert \geq 2 \;\; \text{and} 
\label{brds.01}\\
\sigma_{i} \, \sigma_{i+1} \, \sigma_{i} &= \sigma_{i+1} \, \sigma_{i} \, \sigma_{i+1} \, , &&\text{for} \; \; 1 \leq i \leq n{-}2 \, .
\notag
\end{align}
There exists a canonical projection map $\tau: B_n \rightarrow S_n$ from $B_n$ to the symmetric group $S_n$ that acts~as
\beq
\tau : \sigma_{i, \, i+1} \mapsto (i,i{+}1) \, ,
\eeq
mapping $\sigma_{i,i+1}$ to the transposition $(i,i{+}1)$.

The braid group acts on solutions, ${\cal G}_{n}^\dR$, to the KZ differential equation, \eqref{eq:dkG}. Its elements $\sigma \in B_{n}$ act on ${\cal G}_{n}^\dR$ by the corresponding permutation on the indices of both $z_i$ and~$e_{i,j}$:
\beq\label{eq:permG1}
\sigma {\cal G}^\dR_{n} = \sigma (\GG_{n}^\dR) \, \ldots \, \sigma (\GG_{1}^\dR)\, ,
\eeq
where (noting that the permutations of interest leave $\sigma(0) = 0$ and $\sigma(n{+}1)= n{+}1$ inert)
\beq\label{eq:permG2}
\sigma (\GG_{k}^\dR) = \GG^\dR\left[\begin{matrix} \sum_{j=0}^{k-1} e_{{\sigma(k)}, {\sigma(j)}} & e_{{\sigma(k)}, {\sigma(k+1)}} & \cdots & e_{{\sigma(k)},{\sigma(n)}} & e_{{\sigma(k)},{ n+1}} \\ z_0 & z_{{\sigma(k+1)}} & \cdots & z_{{\sigma(n)}} & z_{n+1} \end{matrix} ; z_{{\sigma(k)}} \right] \, .
\eeq
In slight abuse of notation, we write $\sigma(k)$ for action of the permutation $\tau_\sigma$ on the indices $k=1,\ldots,n$. However, since ${\cal G}^\dR_{n}$ and $\sigma {\cal G}^\dR_{n}$ satisfy the same linear differential equation, they are related to each other by some constant factor:
\begin{equation}\label{eq:braid:action}
\sigma {\cal G}^\dR_n = \BB^\dR(\sigma) {\cal G}^\dR_n
\end{equation}
for some series $\BB(\sigma)$ in the $e_{i,j}$. In particular, the action of a transposition is given by 
\begin{equation}\label{eq:braid:trans}
\BB^\dR(\sigma_{i,i+1}) = \Phi^\dR\left(\sum_{j=0}^{i-1} e_{j,i+1}\, ,\, e_{i,i+1}\right) \Phi^\dR\left(e_{i,i+1}\, , \, \sum_{j=0}^{i-1} e_{j,i}\right)\, ,
\end{equation}
where the analogous formula for $\BB^\mot(\sigma_{i,i+1})$ features an additional factor of $\exp(i\pi e_{i,i+1})$ in between the associators \cite{KZB, Britto:2021prf}. Moreover, it follows from \eqref{eq:braid:action} that a product $\sigma\sigma' \in B_n$ acts on ${\cal G}^\dR_n$ according to
\begin{equation}
\BB^\dR(\sigma \sigma') = \tau_\sigma \big(\BB^\dR(\sigma') \big) \BB^\dR(\sigma) \, ,
\end{equation}
where $\tau_\sigma$ acts on $\BB^\dR(\sigma')$ by permutation of the indices of both $z_i$ and $e_{i,j}$.

\begin{lemma}\label{lem:braid:cycle}
The braid $\sigma_{(a,b)} = \sigma_{a,a+1} \cdots \sigma_{b-1,b}$, which affects the cyclic permutation
\beq
\tau(a,a{+}1,\ldots,b{-}1,b) = (a{+}1,\ldots,b{-}1,b,a)
\eeq
acts on ${\cal G}_n^\dR$ by
\begin{equation}\label{eq:braid:cycle}
\BB^\dR(\sigma_{(a,b)}) = \prod_{i=b-1}^a \Phi^\dR\left( \sum_{\substack{j=0\\j\neq a}}^{i-1} e_{j,i+1}\, , \, e_{a,i+1} \right) \Phi^\dR\left( e_{a,i+1}\, , \, \sum_{\substack{j=0\\j\neq a}}^{i-1} e_{j,a} \right) \, ,
\end{equation}
where the order of multiplication is left-to-right from $i = b{-}1$ to $i=a$. 
\end{lemma}
\begin{proof}
The formula certainly holds for $\sigma = \sigma_{a,a+1}$, by \eqref{eq:braid:trans}. Moreover, $\sigma_{(a-1,b)} = \sigma_{a-1,a} \sigma_{(a,b)}$, such that
\begin{equation}
\label{eq:braid:trans:induction}
\BB^\dR(\sigma_{(a-1,b)}) = \tau_{a-1,a}\big( \BB^\dR(\sigma_{(a,b)})\big)\, \BB^\dR(\sigma_{a-1,a})\, .
\end{equation}
Combining equations \eqref{eq:braid:trans} and \eqref{eq:braid:trans:induction}, the Lemma then follows by induction.
\end{proof}

\subsection{Emergence of the Drinfeld associators in the
product (\ref{simpkey})}
\label{sec:4.3}

We have already seen that the Drinfeld associator arises from our generating functions of MPLs in one variable as the shuffle-regularized limit, for example,
\beq
\lim_{z_n \rightarrow 1} \GG_n[e_{n,0},e_{n,n+1};z_n] = \Phi(e_{n,0},e_{n,n+1})\, .
\eeq
Several other important limits give rise to Drinfeld associators.

\begin{lemma}\label{lem:12limit}
The following shuffle-regularized double limit of $\GG_1$ at arbitrary $n\geq 1$ gives a Drinfeld associator:
\beq\label{eq:12limit}
\lim_{\substack{z_j\rightarrow 0\\ j \geq 2}} \lim_{z_1 \rightarrow z_2} \GG_1[ \{e_{1,i}\}; z_1] = \Phi (e_{1,0},e_{1,2})
\eeq
\end{lemma}
\begin{proof}
In the limit as $z_1 \rightarrow z_2$, the integration kernels in the definition (\ref{coact.01}) of MPLs (to be integrated from $0$ to $z_1$ in case of $\GG_1[\ldots;z_1]$) can be reparameterised by a change of variables from $t$ to $u = t/z_2$, where now $u$ is integrated from $0$ to $1$. The kernels corresponding to $e_{1,0}= e^\ast_{1,0}$ and $e_{1,2}$ become
\beq
\frac{\dd t}{t} = \frac{\dd u}{u} \, ,\qquad \frac{\dd t}{t{-}z_2} = \frac{\dd u}{u{-}1}\, ,
\eeq
respectively. And the kernels corresponding to $e_{1,j}$ (with $j>2$) become
\beq
\frac{\dd t}{t{-}z_j} = \frac{ z_2\, \dd u}{z_2 u {-} z_j}
\eeq
which vanish in the limit $z_2 \rightarrow 0$. In other words, all terms in $\GG_1$ with $e_{1,j}$ ($j>2$) go to zero. It follows that, in the limit of \eqref{eq:12limit}, $\GG_1$ becomes the generating series for MZVs, with $\dd u/u$ and $\dd u/(1{-}u)$ accompanied by $e_{1,0}$ and $e_{1,2}$, respectively: this is precisely $\Phi(e_{1,0},e_{1,2})$.
\end{proof}

For the next Lemma, we introduce particular linear combinations of the braid generators, $E_0^{(r)}$ and $E_1^{(r)}$, for $r = 1,2, \ldots, 2n{-}1$. These are given by
\begin{align}
E_0^{(2a-1)} &= e_{1,0}+\sum_{i=2}^a e_{1,i}\, , & E_0^{(2a)} & = e_{1,a+1}  \, ,\label{eq:E0E1} \\
E_1^{(2a-1)} &= e_{1,a+1} \, ,& E_1^{(2a)} &= e_{0,a+1}+\sum_{i=2}^a e_{i,a+1} \, .\notag  
\end{align}
These linear combinations arise from the braid action, Lemma \ref{lem:braid:cycle}, for braids of the form $\sigma_{(1,\ell)}$. Using these we find the following lemma on more general shuffle-regularized double limits.

\begin{lemma}\label{lem:bigphiprod}
Fix some $\ell$ (with $2\leq \ell \leq n{+}1$). The following shuffle-regularized limit is given~by
\beq
\lim_{\substack{z_i\rightarrow 0\\ i \geq 2}} \YY^\dR_\ell = \lim_{\substack{z_i\rightarrow 0\\ i\geq 2}} \, \lim_{z_{1} \rightarrow z_{\ell}} \, {\cal G}^\dR_{n} = \Phi^{(1)} \Phi^{(2)} \cdots \Phi^{(2\ell-3)}\, ,
\label{keyrst}
\eeq
where $\Phi^{(r)} = \Phi^\dR(E_{0}^{(r)}, \, E_{1}^{(r)})$.
\end{lemma}

\begin{proof}
For $\ell = 2$ we can apply Lemma \ref{lem:12limit} to find
\beq
\lim_{\substack{z_i\rightarrow 0\\ i\geq 2}} \, \lim_{z_{1} \rightarrow z_{2}} \, {\cal G}^\dR_{n} = \Phi^\dR\big(E_0^{(1)}, E_1^{(1)}\big) = \Phi^{(1)}
\eeq
since $E_0^{(1)} = e_{1,0}$ and $E_1^{(1)} = e_{1,2}$ by (\ref{eq:E0E1}). Here we have used that
\beq
\lim_{\substack{z_i\rightarrow 0\\ i\geq 2}} \, \lim_{z_{1} \rightarrow z_{2}} \GG_m = 1
\eeq
for all $m\geq 2$, since $\GG_m$ only depends on $z_m,z_{m+1},\ldots,z_n$ and is thus
unaffected by the inner limit $z_1\rightarrow z_2$. The outer limits $z_i\rightarrow 0$ for all $i\geq 2$ then shrink the integration domain to zero size.

Fix $\ell > 2$. Taking the $z_i \rightarrow 0$ limit of  $\lim_{z_1\rightarrow z_\ell} \GG_1$ is difficult to do directly since $z_1$ and $z_{\ell}$ are not adjacent in the ordering (\ref{eq:zorder}) prescribed for real values of $z_i$. However, we can use the braid action to apply Lemma \ref{lem:12limit} also in this case. We use the cycle braid $ \sigma_{(1,\ell-1)} = \sigma_{1,2} \sigma_{2,3}\ldots \sigma_{\ell-2,\ell-1}$ to gradually move $z_1$ to be adjacent to $z_\ell$. In fact, in Appendix~\ref{sec:D}, we show that any choice of~braid that implements the cyclic permutation of $1,2,\ldots,\ell{-}1$ gives rise to the same result. By Lemma~\ref{lem:braid:cycle},
\beq \label{eq:bigphi1}
{\cal G}^\dR_{n} = \BB^\dR \big(\sigma_{(1,\ell-1)} \big)^{-1}\, \sigma_{(1,\ell-1)} \, {\cal G}^\dR_{n} \, ,
\eeq
where we emphasize that $\BB^\dR (\sigma_{(1,\ell-1)})$ does not depend on the variables $z_i$. The permuted generating series, $\sigma_{(1,\ell-1)} \, {\cal G}^\dR_{n}$, is given by (see \eqref{eq:permG1} and \eqref{eq:permG2})
\begin{equation}
\sigma_{(1,\ell-1)} \, {\cal G}^\dR_{n} = \prod_{k=n}^{1} \GG^\dR_k \big[\tau (e_{k,0}^{\ast}), \{ \tau(e_{k,r})\}; \tau(z_{k}) \big]  \, ,
\end{equation}
where $\tau$ acts as the cyclic permutation $\tau(1, 2, \, \ldots, \, \ell{-}1) = (2,3,\ldots,\ell{-}1,1)$ on the indices\footnote{In the case of $e_{k,0}^\ast$, the permutation $\tau$ acts on the indices of each $e_{i,j}$ appearing in the sum $e_{k,0}^\ast = \sum_{i=0}^{k-1}e_{k,i}$.} and the product $\prod_{k=n}^{1} $ is performed in descending order, i.e.\ $\GG^\dR_n\ldots \GG^\dR_2 \GG^\dR_1$. For $k \neq \ell{-}1$,
\beq
\lim_{\substack{z_i\rightarrow 0\\ i\geq 2}} \, \lim_{z_{1} \rightarrow z_{\ell}} \GG^\dR_k \big[\tau (e_{k,0}^{\ast}), \{ \tau(e_{k,r})\}; \tau(z_{k}) \big]  = 1
\eeq
since $\tau(k) \neq 1$. For $k=\ell{-}1$, however, we have $\tau(k)=1$, and $z_1$ is now adjacent to $z_\ell$ in the new ordering. The shuffle-regularized limit can be computed using the same method as Lemma \ref{lem:12limit}:
\beq\label{eq:bigphi2}
\lim_{\substack{z_i\rightarrow 0\\ i\geq 2}} \, \lim_{z_{1} \rightarrow z_{\ell}} \GG_{\ell-1} \big[\tau(e_{k,0}^\ast), \{ e_{1,\tau(r)}\}; z_1 \big] = \Phi\big(E_0^{(2\ell-3)},E_1^{(2\ell-3)}\big)\, ,
\eeq
where we have used that $E_0^{(2\ell-3)} = \tau(e_{k,0}^\ast) = e_{1,0}{+} \sum_{i=2}^{\ell-1}e_{1,i}$ and $E_1^{(2\ell-3)} = e_{1,\ell}$ by (\ref{eq:E0E1}).

Finally, $\BB^\dR (\sigma_{(1,\ell-1)})$ is obtained from the general formula \eqref{eq:braid:cycle} for cycle braids with inverse
\begin{align}
\BB^\dR (\sigma_{(1,\ell-1)})^{-1} &= \prod_{i=1}^{\ell-2} \Phi^\dR\left( \sum_{\substack{j=0\\j\neq 1}}^{i-1} e_{j,1}\, , \, e_{1,i+1} \right) \Phi^\dR\left( e_{1,i+1} \, , \, \sum_{\substack{j=0\\j\neq 1}}^{i-1} e_{j,i+1}\right)  \notag \\
&= \prod_{r=1}^{2\ell-4} \Phi^\dR \big(E_{0}^{(r)}, E_{1}^{(r)} \big) = \Phi^{(1)} \Phi^{(2)} \ldots \Phi^{(2\ell-4)} \, ,
\label{eq:bigphi3}
\end{align}
where we use the identity $\Phi(e_0,e_1)^{-1} = \Phi(e_1,e_0)$ (equation (\ref{eq:PhiPhi1})). As one can see from the last step, we take the products here as ordering their factors left-to-right with increasing $r$. The Lemma follows by multiplying \eqref{eq:bigphi3} and \eqref{eq:bigphi2}, using \eqref{eq:bigphi1} and identifying $E_{0}^{(r)}, E_{1}^{(r)}$ according to (\ref{eq:E0E1}).
\end{proof}
By the $z_i$ independence of $\YY^\dR_\ell$ shown in Lemma \ref{lem:YeY}, the statement (\ref{keyrst}) of the Lemma implies the key step (\ref{simpkey}) towards proving Theorem \ref{thm:main}.

\section{The coaction formula for MPLs of a single variable}
\label{sec:5}
The previous section completed steps (i) and (ii) of Section \ref{sec:3.3} towards the proof of Theorem~\ref{thm:main}. Before proceeding to step (iii) of the proof of Theorem \ref{thm:main} for any number of variables, $n$, we introduce the key ideas and lemmas by studying the special case of $n=1$.

\subsection{Expansion of the Drinfeld associator}
\label{sec:5.1}
The proof of Theorem \ref{thm:main} relies on some properties of the Drinfeld associators, $\Phi^\dR$ and $\Phi^\mot$, that we prove in this section. These series can be written as
\begin{equation}
\Phi^\dR(e_0,e_1) = \sum_W \phi^{-1}(f^\dR_W) P_W,\qquad \Phi^\mot(e_0,e_1) = \sum_{k=0}^\infty\sum_W \phi^{-1}\big((f_2^\mot)^kf^\mot_W \big) P_{(2k)W} \, ,
\label{expphimdr}
\end{equation}
upon converting the de Rham and motivic MZVs into the $f$-alphabet as in Section \ref{sec:2.3}, where $P_W$ and $P_{(2k)W}$ are polynomials in $e_0,e_1$. The parenthesis of $P_{(2k)W}$ ensures that the integer $2k$ is treated as a single letter, and we have $P_{(0)W}= P_{W}$. The sums over $W$ are over all words in odd numbers $\geq 3$ with $f_W=f_a f_b\ldots f_c$ for $W=ab\ldots c$. The choice of an isomorphism $\phi$ affects the expressions for the polynomials $P_W, P_{(2k)W}$. A recent proposal for canonical choices of such polynomials and the $\phi$ isomorphism can be found in \cite{Dorigoni:2024iyt}, though our main results including Theorems \ref{thm:main} and \ref{thm:sv} are unaffected by these choices.

In fact, the Drinfeld associators are group-like (as understood in the original papers by Drinfeld \cite{Drinfeld:1989st, Drinfeld2}, but see also \cite{Bar:1998, BrownTate}), which means that
\beq\label{eq:grouplike}
\delta_\shuffle \Phi = \Phi \otimes \Phi,
\eeq
where $\delta_\shuffle$ can be explicitly written as the de-shuffle coproduct on this associative algebra\footnote{In most references, the group-like property is expressed in terms of the coproduct on the universal enveloping algebra of the free Lie algebra and is defined as $\delta P = P \otimes 1 + 1 \otimes P$ for any element $P$ of the free Lie algebra. But we can identify $\delta$ with $\delta_\shuffle$ by using that the universal enveloping algebra is isomorphic to the free associative algebra on $e_0$ and $e_1$. This follows from the Poincar\'e--Birkoff--Witt Theorem, see for instance \cite{Reutenauer}.}
\beq
\delta_\shuffle e_A = \sum_{B,C} (A,B\shuffle C) e_B \otimes e_C
\label{shcoprod}
\eeq
for words $A,B,C$ in $0$ and $1$, where $e_A = e_a e_b \ldots e_c$ for a word $A=ab\ldots c$. Here we sum over all words $B,C$ including empty ones. The inner product $(A, B\shuffle C)$ picks out those $B,C$ such that $A$ appears in the shuffle product $B\shuffle C$. An important property of $\delta_\shuffle$ is that a polynomial $P$ in the free associative algebra is a Lie polynomial iff
\beq
\delta_\shuffle P = P\otimes 1 + 1 \otimes P\, ,
\eeq
see Appendix \ref{app:Lie} for a review of free Lie algebras and their properties.

\begin{lemma}
The polynomials $P_W$ and $P_{(2k)W}$ appearing in the expressions (\ref{expphimdr}) for $\Phi^\dR$ and $\Phi^\mot$ satisfy
\begin{equation}\label{eqn:deshufflePW}
\delta_\shuffle P_W = \sum_{E,F} (W,E\shuffle F) P_E \otimes P_F
\end{equation}
and
\beq\label{eqn:deshuffleP2W}
\delta_\shuffle P_{(2k)W} = \sum_{k_1+k_2 = k} \sum_{E,F} (W,E\shuffle F) P_{(2k_1)E} \otimes P_{(2k_2)F}\, ,
\eeq
where $W,E,F \in {\cal W}$ are words in the odd numbers (i.e. $3,5,7,\ldots$), and $k,k_1,k_2 \geq 0$.
\end{lemma}
\begin{proof}
On the one hand, the de-shuffle coproduct of the Drinfeld associator can be carried out at the level of
\beq
\delta_\shuffle\Phi^\dR = \sum_W \phi^{-1}(f^\dR_W) \delta_\shuffle P_W \, .
\label{matpart1}
\eeq
On the other hand, the group-like property \eqref{eq:grouplike} implies that
\beq
\Phi^\dR\otimes \Phi^\dR = \sum_{E,F} \phi^{-1}(f^\dR_E)\phi^{-1}(f^\dR_F)  P_E\otimes P_F = \sum_{E,F} \phi^{-1}(f^\dR_E \shuffle f^\dR_F)  P_E\otimes P_F\, ,
\label{matpart2}
\eeq
where we used the shuffle product of MZVs (see equation \eqref{eq:fshuffle} in Section \ref{sec:2.3}). Matching the coefficients of $\phi^{-1}(f^\dR_W)$ on the right-hand sides of (\ref{matpart1}) and (\ref{matpart2}) implies the first part \eqref{eqn:deshufflePW} of the Lemma. The second part of the Lemma follows from the analogous calculation for $\Phi^\mot$, because the powers of $f_2$ appearing in $\Phi^\mot$ do not affect the shuffle product of the odd generators $f_{2k+1}$ in the $f$-alphabet (see (\ref{shuffprod}) in Section \ref{sec:2.3}).
\end{proof}

\begin{lemma}\label{lem:liecoef}
The polynomials $P_{(2)}$ and $P_m$ (for $m=3,5,\ldots$) that multiply the Riemann zeta values $\phi^{-1}(f_m)$ in $\Phi^\dR$ and $\Phi^\mot$ are Lie polynomials.
\end{lemma}
\begin{proof}
By \eqref{eqn:deshufflePW} and \eqref{eqn:deshuffleP2W}, the polynomials $P_m$ for $m=3,5,\ldots$ in both $\Phi^\dR$ and $\Phi^\mot$ satisfy $\delta_\shuffle P_m = P_m\otimes 1 + 1 \otimes P_m$, which implies that $P_m$ is a Lie polynomial. Similarly, \eqref{eqn:deshuffleP2W} implies that the polynomial $P_{(2)} = [e_0,e_1]$ appearing in $\Phi^\mot$ satisfies $\delta_\shuffle P_{(2)} = P_{(2)}\otimes 1 + 1 \otimes P_{(2)}$, so this is also a Lie polynomial.
\end{proof}

\begin{remark}
Note that the polynomials $P_{(2k)}$, for $k>1$, that multiply the even Riemann zeta values $\zeta^\mot_{2k}$ in $\Phi^\mot$ with $2k\geq 4$ are \emph{not} Lie polynomials.\footnote{Note that the canonical polynomials $g_k$ of \cite{Dorigoni:2024iyt} with even $k \geq 4$ follow normalization conventions adapted to the element $\zeta_k$ in the $\mathbb Q$ basis of (motivic) MZVs such that the term $g_k = {\rm ad}_x^{k-1}(y)+\ldots$ with the maximum number of letters $x$ appears with unit coefficient. The corresponding $P_{(k)}$ with even $k \geq 4$ in this work, by contrast, are normalized to feature $ {\rm ad}_{e_0}^{k-1}(e_1)$ with coefficients $\zeta_k / (\zeta_2)^{k/2}$, reflecting a $\mathbb Q$ basis of (motivic) MZVs including $(\zeta_2)^{k/2}$ instead of $\zeta_k$.} For example, $P_{(4)}$ satisfies 
\beq
\delta_\shuffle P_{(4)} = P_{(4)} \otimes 1 + P_{(2)} \otimes P_{(2)} + 1 \otimes P_{(4)}\, ,
\eeq
so is not Lie. Indeed $P_{(4)}$ is given by \cite{LeMura}
\beq
P_{(4)} = \frac{2}{5} [e_0,[e_0,[e_0,e_1]]] + \frac{1}{10} [e_1,[e_0,[e_1,e_0]]] -  \frac{2}{5} [e_1,[e_1,[e_1,e_0]]] +  \frac{1}{2} [e_0,e_1] [e_0,e_1]
\eeq
whose last term $\sim [e_0,e_1] [e_0,e_1]$ is manifestly not a Lie polynomial.
\end{remark}

\subsubsection{Ihara product}
The motivic coaction on $\Phi^\mot(e_0,e_1)$ can be computed in two different ways. First, it is induced by the coaction defined on the individual motivic MZVs, as reviewed in Section \ref{sec:2.3} (see \eqref{eq:fcoaction} and \eqref{eq:fdelta}). Second, the Ihara formula also gives a formula for $\Delta \Phi^\mot(e_0,e_1)$, as reviewed in Section \ref{sec:iharaMZV} (see \eqref{eq:iharaphi}). Combining these two formulas for $\Delta \Phi^\mot(e_0,e_1)$ implies some useful properties of the polynomials $P_W, P_{(2k)W}$ appearing in $\Phi^\dR$ and $\Phi^\mot$.

To state these properties, it is helpful to define the \emph{Ihara derivation} on the free Lie algebra generated by $e_0,e_1$. For any Lie polynomial $x$, the derivation $D_x$ is defined by \cite{Ihara1989TheGR, Ihara:stable}
\begin{equation}\label{eq:De0De1}
D_x e_0 = 0\, , \qquad \qquad D_x e_1 = [e_1,x]
\end{equation}
and satisfies the Leibniz property with respect to the Lie bracket
\beq\label{eq:DLeib1}
D_x ([y,z]) = \left[ D_xy, z\right] + \left[y, D_xz \right] \, .
\eeq
The \emph{Ihara product} is then defined by
\begin{equation}\label{eq:circ}
x \circ y = xy - D_y x\, .
\end{equation}
We emphasize that $x\circ y$ is defined for a polynomial $x$ (in $e_0$, $e_1$) and a Lie polynomial $y$. The product $P\circ P'$, for two Lie polynomials, is not itself a Lie polynomial. So repeated applications of $\circ$ to Lie polynomials $P_a,P_b,P_c,\ldots,P_d$ can only be taken in the form
\beq
(\cdots ((P_{a} \circ P_b) \circ P_c) \circ \cdots )\circ P_{d} \, ,
\eeq
by bracketing on the left.

The statement of the following Lemma is equivalent to an observation in \cite{Drummond:2013vz}, and we shall present a proof according to \cite{FB:talk, priv:FB} (also see \cite{priv:LS} for an alternative proof). 

\begin{lemma}\label{lem:drinexp}
The polynomials $P_W$ and $P_{(2k)W}$ appearing in $\Phi^\dR$ and $\Phi^\mot$ in \eqref{expphimdr} satisfy ($a\geq 3$~odd)
\beq
P_{Wa} = P_W \circ P_a\, ,\qquad P_{(2k)Wa} = P_{(2k)W} \circ P_a\, .
\label{stlm54}
\eeq
In particular, this implies the relations
\begin{align}
P_W & = (\cdots ((P_{a} \circ P_b) \circ P_c) \circ \cdots )\circ P_{d}\, , \label{stlm54b}\\
P_{(2k)W} & = (\cdots (((P_{(2k)} \circ P_{a}) \circ P_b) \circ P_c) \circ \cdots )\circ P_{d} \, , \notag
\end{align}
for words $W = abc\cdots d$ in odd letters ($a,b,\ldots,d\geq 3$), and where, by Lemma \ref{lem:liecoef}, the $P_a$ are Lie polynomials.
\end{lemma}
\begin{proof}
At leading orders w.r.t.\ the number of letters $f_a$ (or \textit{coradical degree}) in the expansion of the expression \eqref{eq:eprimephi} for $e_1'$,
\begin{equation}
e_1' = e_1 - \sum_a \phi^{-1}(f^\dR_a) [e_1,P_a] + \cdots\, .
\end{equation}
Using this, the factor, $\Phi^\mot(e_0,e_1')$, in the Ihara formula can likewise be expanded
\begin{equation}\label{eq:phimotexp}
\Phi^\mot(e_0,e_1') = \Phi^\mot(e_0,e_1) - \sum_a \phi^{-1}(f^\dR_a) D_{P_a} \Phi^\mot(e_0,e_1) + \cdots\, ,
\end{equation}
where each term in the ellipsis features at least two letters $f^\dR_a f^\dR_b$ with $a,b$ odd.
For any word $A$ in $3,5,7,\ldots$ and letter $a$, consider the coefficient of $f_A^\mot f_a^\dR$ in the Ihara formula, \eqref{eq:iharaphi}. The left-hand side can be computed using the motivic coaction (\ref{eq:fcoaction}) on the $f$-alphabet. This implies that the coefficient of $f_A^\mot f_a^\dR$ on the left-hand side is $P_{Aa}$. Whereas, on the right-hand side, the coefficient of $f_A^\mot f_a^\dR$ follows from \eqref{eq:phimotexp}, and is given by
\begin{equation}
P_AP_a - D_{P_a}P_A = P_A \circ P_a
\end{equation}
such that
\begin{equation}
P_{Aa} = P_A \circ P_a \, .
\end{equation}
The argument can be straightforwardly extended to the coefficient of $(f_2^\mot)^kf_A^\mot f_a^\dR$ in \eqref{eq:iharaphi} with $k\geq 1$ which is $P_{(2k)Aa}$ on the left-hand side and $P_{(2k)A} \circ P_a,$ on the right-hand side. In view of $P_{(2k)Aa}=P_{(2k)A} \circ P_a$ and 
$P_{(0)A}=P_{A}$, this implies the
statement (\ref{stlm54b}) of the Lemma.
\end{proof}

To describe the inverse series $\Phi^\dR(e_0,e_1)^{-1}$, it is helpful to introduce a new `dual' Ihara product for words $y$ and Lie polynomials $x$ in $e_0, e_1$,
\beq
x{\,\widetilde\circ\,} y = xy + D_x y \, .
\eeq
We define iterated products of Lie polynomials $P_a$ with $\widetilde\circ$ as
\beq
\widetilde{P}_A = P_{a} {\,\widetilde\circ\,} \big(\cdots {\,\widetilde\circ\,} (P_b {\,\widetilde\circ\,} (P_c {\,\widetilde\circ\,} P_d) )\cdots \big)\, ,
\label{nottildeB}
\eeq
for a word $A = a \cdots bcd$, where now we bracket the product on the right.

We will also employ the antipode $\alpha(A)$ for words $A$ in arbitrary alphabets defined by
\beq\label{eq:antipodeA}
\alpha(A) = (-1)^{|A|} A^t\, ,
\eeq
where $|A|$ is the length of the word $A$ and $A^t$ is the reversed word (see \eqref{lie:antipode1} in Appendix \ref{app:Lie}).

\begin{lemma}\label{lem:PhiInv}
The inverse of the series expansion (\ref{expphimdr}) of $\Phi^\dR(e_0,e_1)$ is given by
\beq\label{eq:PhiInv}
\Phi^\dR(e_0,e_1)^{-1} = \sum_{W} \phi^{-1}(f^\dR_{\alpha(W)}) \widetilde{P}_W \, ,
\eeq
where $\widetilde{P}_W$ is defined by \eqref{nottildeB}.
\end{lemma}

\begin{proof}
Equation \eqref{eq:PhiInv} is equivalent to showing (for nonempty words $W$) that
\beq\label{eq:AcircB}
\sum_{A,B} \big(W,A\shuffle \alpha(B) \big) P_A \widetilde{P}_{B} = 0\, .
\eeq
Clearly this holds when $|W|=1$. In order to establish the validity of (\ref{eq:AcircB}) for words of length $|W| \geq 2$, it is convenient to act on its left-hand side with $D_{P_i}$, which gives
\begin{align}
&\sum_{A,B} \big(W,A\shuffle \alpha(B) \big) \big[ (-P_A\circ P_i + P_A P_i) \widetilde{P}_{B} + P_A (P_i {\,\widetilde\circ\,} \widetilde{P}_B - P_i \widetilde{P}_B) \big] \notag \\
&= \sum_{A,B} \big(W,A\shuffle \alpha(B) \big) \big[  P_A (P_i {\,\widetilde\circ\,} \widetilde{P}_B )
- (P_A\circ P_i ) \widetilde{P}_{B}\big] \notag \\
&= \sum_{A,B} \big(W,A\shuffle \alpha(B) \big) \big[  P_A \widetilde{P}_{iB} 
- P_{Ai} \widetilde{P}_{B}\big]
\, ,
\label{eq:PhiInv-1}
\end{align}
using Lemma \ref{lem:drinexp} in passing to the last line. Now we recall (see \eqref{lie:shuffle} in Appendix \ref{app:Lie})
\begin{equation}\label{eq:lemmashuffle}
Aa \shuffle Bb = (A \shuffle Bb)a + (Aa \shuffle B)b
\end{equation}
so that
\beq
\sum_{A,B} \big(W i ,A\shuffle \alpha(B) \big) P_A \widetilde{P}_B = \sum_{A,B} \big(W ,A\shuffle \alpha(B) \big) \left[ P_{Ai} \widetilde{P}_B - P_A \widetilde{P}_{iB} \right] \, .
\eeq
Comparing with \eqref{eq:PhiInv-1}, 
we find
\beq
D_{P_i}\, \sum_{A,B} \big(W,A\shuffle \alpha(B) \big) P_A \widetilde{P}_{B} = - \sum_{A,B}  \big(W i ,A\shuffle \alpha(B)\big) P_A \widetilde{P}_B
\eeq
such that \eqref{eq:AcircB} follows by induction in the length of the word $W$.
\end{proof}

\begin{remark}
Equivalently, Lemma \ref{lem:PhiInv} follows from the identity $\Phi^\dR(e_0,e_1)^{-1} = \Phi^\dR (e_1,e_0)$, see  \eqref{eq:PhiPhi1}. For the coefficient of $\phi^{-1}(f_a)$, this implies that $P_a(e_1,e_0) = - P_a(e_0,e_1)$. Moreover, let $\tau$ denote the swap $e_0 \leftrightarrow e_1$. By the definition of the Ihara derivation, (\ref{eq:De0De1}),
\beq
D_y x + \tau \left( D_{\tau y} \tau x\right) = [x,y] \, .
\eeq
This implies that,
\beq
\tau (x \circ y) = (\tau y) {\,\widetilde\circ\,} (\tau x) \, .
\eeq
So, using $\tau P_a = - P_a$, we have $\tau P_W = \widetilde P_{\alpha(W)}$, which again gives the Lemma.
\end{remark}

\begin{lemma}\label{lem:PhiXPhi}
Conjugating $e_1$ by $\Phi(e_0,e_1)$ gives the series
\begin{equation}
\Phi^{\dR}(e_0,e_1) \, e_1 \left( \Phi^{\dR}(e_0,e_1)\right)^{-1} = \sum_{W=ab\cdots c} \phi^{-1}\big(f_{\alpha(W)}^\dR\big) D_{a} D_b \cdots D_c e_1\, ,
\end{equation}
where $D_a = D_{P_a}$.
\end{lemma}

\begin{proof}
Lemma \ref{lem:PhiInv} implies that
\beq
\Phi^{\dR}(e_0,e_1) \, e_1 \big( \Phi^{\dR}(e_0,e_1)\big)^{-1} = \sum_W \phi^{-1}(f^\dR_W) \sum_{A,B} \big(W,A\shuffle \alpha(B) \big) P_A e_1 \widetilde{P}_B\, .
\eeq
So Lemma \ref{lem:PhiXPhi} then follows if we can show that (for a word $W = ab \cdots c \in {\cal W}$)
\beq\label{eq:AcircBD}
\sum_{A,B} \big(W,A\shuffle \alpha(B) \big) P_A e_1 \widetilde{P}_B = (-1)^{|W|} D_{c}\cdots D_b D_a e_1\, .
\eeq
But clearly \eqref{eq:AcircBD} holds for $|W|=1$. Its validity at length $|W|\geq 2$ can be established by acting on the left-hand side of \eqref{eq:AcircBD} with $D_i=D_{P_i}$ which gives
\begin{align}
D_i \big( P_A e_1 \widetilde{P}_B \big) &= - (P_A P_i - D_i P_A) e_1 P_B + P_A e_1 (P_i \widetilde{P}_B + D_i \widetilde{P}_B) \notag \\
&= - P_{Ai} e_1 P_B + P_A e_1 \widetilde{P}_{iB}\, ,
\label{dipa1b}
\end{align}
where we recall that $D_i e_1 = [e_1,P_i]$. Applying the identity \eqref{eq:lemmashuffle} once again, \eqref{eq:AcircBD} then follows by induction, analogous to the proof of Lemma \ref{lem:PhiInv}.
\end{proof}

\subsection{Properties of the zeta generating series}
\label{sec:5.2}

The formal MZV generating series, $\MM$, defined in Section \ref{sec:2.3} is group-like.

\begin{lemma}\label{lem:Minv}
An inverse of $\MM$ in (\ref{defmser}) is given by the series
\begin{equation}
\MM^{-1} = \sum_{W} 
\phi^{-1} (f_{\alpha(W)}) M_{W} \, ,
\end{equation}
where we sum over all words $W$ in the odd numbers ($3,5,7,\ldots$).
\end{lemma}
\begin{proof}
By the antipode identity (see equation (\ref{lie:antipode2}) in Appendix \ref{app:Lie})
\begin{equation}\label{eq:Minvlem}
\sum_{W=AB} \phi^{-1} (f_{\alpha(A)}) \phi^{-1} (f_B) = \sum_{W=AB} \phi^{-1} (f_{\alpha(A)\shuffle B}) = 0
\end{equation}
for $W\neq \emptyset$. It follows that $\MM^{-1}\MM = \MM\MM^{-1} = 1$. 
\end{proof}

It will be helpful to recall the following notation from Appendix \ref{app:Lie}:

\begin{definition}
\label{deftlbr}
Write
\beq
\ell[a,b,c,\ldots,d] = [ [\ldots [[a,b],c], \ldots ],d]
\label{totlbrk}
\eeq
for the total left-bracketing of a word $abc\ldots d$ in some symbols. 
\end{definition}
The total left-bracketing (\ref{totlbrk}) of a word, $aW$, can be expanded as
\begin{equation}\label{eq:leftbracketing}
\ell[a,W] = \sum_{A,B} \big(W,\alpha(A)\shuffle B \big) AaB
\end{equation}
for a letter, $a$, and word $W$.

\begin{lemma}\label{lem:MXM}
Let $X$ be a polynomial in the zeta generators $M_k$ or the braid generators $e_0, e_1$. Then conjugating $X$ by $\MM$ gives the series
\begin{equation}
\MM^{-1} X \MM = \sum_W \phi^{-1}(f_W) \, \ell[X, M_W] \, .
\end{equation}
\end{lemma}
\begin{proof}
We again use equation \eqref{eq:Minvlem} from the previous Lemma to write
\beq
\MM^{-1} X \MM = \sum_{A,B} \phi^{-1} (f_{\alpha(A)\shuffle B})  \, M_A X M_B\, .
\eeq
But by the expansion of the left-bracketing, \eqref{eq:leftbracketing}, this gives the Lemma.
\end{proof}

\subsection{Towards proving Theorem \ref{thm:main}: the single-variable case} 
\label{sec:5.3}

We now use the results of Sections \ref{sec:5.1} and \ref{sec:5.2} to prove our coaction formula (\ref{thm2.1eq}) for the special case of MPLs in one variable, which is a crucial step towards the inductive proof of Theorem \ref{thm:main}.

\begin{lemma}\label{lem:conjugation_one}

Define an adjoint action of the $M_a$ generators on the braid generators by
\beq \label{eq:e1Ma}
[e_0,M_a]=0  \, , \qquad
[e_1,M_a] = - [e_1,P_a]\, , 
\eeq
where we recall that $P_a = \left.\Phi(e_0,e_1)\right|_{\phi^{-1}(f_a)}$ are the Lie polynomials in ${\rm Lie}\langle e_0,e_1\rangle$ appearing in the expansion (\ref{expphimdr}) of the Drinfeld associator. Then, the following two conjugations of $e_1$ are identical:
\begin{equation}
\left( \MM^\dR\right)^{-1}e_1\MM^\dR = \Phi^{\dR}(e_0,e_1) \,e_1 \big( \Phi^{\dR}(e_0,e_1)\big)^{-1}\, .
\label{conjMphi}
\end{equation}
\end{lemma}
\begin{proof}
By Lemma \ref{lem:MXM}, the left-hand side has the following expansion in MZVs,
\begin{equation}
\left( \MM^\dR\right)^{-1}e_1\MM^\dR = \sum_{W=ab\cdots c} \phi^{-1}\big(f_W^\dR\big) \ell[e_1,M_a,M_b,\ldots,M_c]\, ,
\end{equation}
where $\ell[e_1,M_a,M_b,\ldots] = [\ldots[[e_1,M_a],M_b],\ldots]$ is the total left bracketing. By swapping all the brackets (i.e.\ $[e_1, M_a] = - [M_a, e_1]$, etc.), we can express this instead in terms of right bracketings:
\begin{equation}\label{proof:conj:MM}
\left( \MM^\dR\right)^{-1}e_1\MM^\dR = \sum_{W=ab\cdots c} (-1)^{|W|}\, \phi^{-1}\big(f_W^\dR\big)  [M_c, [\cdots , [M_b, [M_a, e_1]]\cdots ]]
\end{equation}
On the other hand, the right-hand side has the following expansion in MZVs, by Lemma \ref{lem:PhiXPhi},
\begin{equation}\label{proof:conj:DD}
\Phi^{\dR}(e_0,e_1) \, e_1 \big( \Phi^{\dR}(e_0,e_1)\big)^{-1} = \sum_{W=ab\cdots c} (-1)^{|W|}\, \phi^{-1}\big(f_{W}^\dR\big) D_{P_c}  \cdots D_{P_b} D_{P_a} e_1\, ,
\end{equation}
where the derivations $D_{P_i}$ are defined by \eqref{eq:De0De1}. The definition of the adjoint action, \eqref{eq:e1Ma}, then implies that the coefficients of any given MZV, $\phi^{-1}(f_W^\dR)$, are equal in the two series. Indeed one can deduce $ [M_c, [\cdots , [M_b, [M_a, e_1]]\cdots ]] = D_{P_c}  \cdots D_{P_b} D_{P_a} e_1$ from
\beq\label{proof:conj:MD}
[M_a, e_1 ]  = D_{P_a} e_1 = [e_1, P_a] \, ,\qquad\text{and} \qquad [M_a, e_0] = D_{P_a} e_0 = 0\, .
\eeq
But both $[M_a,\cdot]$ and $D_{P_a}$ are derivations on the free Lie algebra generated by $e_0, e_1$.\footnote{For any $[x,y]$ in the free Lie algebra, $D_P [x,y] = [D_P x, y] + [x, D_P y]$ and $[M_a,[x,y]] = [[M_a,x],y] + [x,[M_a,y]]$, by the Jacobi identity among nested brackets.} So \eqref{proof:conj:MD} implies that $[M_a, P] = D_{P_a} P$ for any Lie polynomial, $P$. In particular, the nested bracketings of $[M_a,\cdot]$ appearing in the series \eqref{proof:conj:MM} are equal to the nested applications of $D_{P_a}$ appearing in the series \eqref{proof:conj:DD}.
\end{proof}

\begin{prop}
\label{prop:1thm}
Recall the adjoint action of the $M_a$ generators on $e_0,e_1$ in (\ref{eq:e1Ma}). Then,
\begin{equation}
\Delta \GG^\mot[e_0,e_1;z] = \left(\MM^{\dR}\right)^{-1} \GG^\mot[e_0,e_1;z] \,\MM^\dR \GG^\dR[e_0,e_1;z] \, .
\end{equation}
\end{prop}
\begin{proof}
Since the zeta generators $M_a$ commute with $e_0$,
\begin{equation}\label{eq:1thmproof1}
\left(\MM^{\dR}\right)^{-1} \GG^\mot [e_0,e_1;z] \,\MM^\dR = \GG^\mot \Big[e_0,\left( \MM^\dR\right)^{-1}e_1\MM^\dR;z\Big] \, .
\end{equation}
Moreover, we recall the Ihara formula for $n=1$ (equation \eqref{eq:1ihara} in Section \ref{sec:iharaMZV})
\begin{equation}\label{eq:1thmproof2}
\Delta \GG^\mot[e_0,e_1;z] = \GG^\mot[e_0,e_1';z]\GG^\dR[e_0,e_1;z] \, ,
\end{equation}
where
\begin{equation}
e_1' = \Phi^\dR(e_0,e_1)\,e_1\,\Phi^\dR(e_0,e_1)^{-1} \, .
\end{equation}
Comparing \eqref{eq:1thmproof2} with \eqref{eq:1thmproof1}, we see that the Proposition follows as a consequence of the Ihara formula and the statement (\ref{conjMphi}) of Lemma \ref{lem:conjugation_one}.
\end{proof}

\begin{remark} 
The computation in Lemma \ref{lem:conjugation_one} for length-one words $W=a$ could have been used to \emph{discover} the definition of the adjoint action, (\ref{eq:e1Ma}), that is needed for Proposition \ref{prop:1thm} to hold. Hence, the non-trivial achievement in this section is to demonstrate that the choice of $[e_i,M_a]$ in (\ref{eq:e1Ma}) dictated by $|W|=1$ is sufficient for (\ref{conjMphi}) to hold at arbitrary $|W|\geq 2$.
\end{remark}

Indeed, the Ihara formula, \eqref{eq:1thmproof2}, plays \emph{two} roles in the proof of Proposition \ref{prop:1thm}. The first role is to see directly that Lemma \ref{lem:conjugation_one} above implies the Proposition by using the Ihara formula for the coaction of MPLs. The second role is more indirect: we use the Ihara formula for the coaction of MZVs in Section \ref{sec:5} to prove the formula for the Drinfeld associator in Lemma \ref{lem:drinexp}. This is what leads, via Lemma \ref{lem:PhiXPhi}, to the above Lemma \ref{lem:conjugation_one}. 

\begin{remark}
\label{imprmk}
The Lie polynomials $P_a = \left.\Phi(e_0,e_1)\right|_{\phi^{-1}(f_a)}$  of degrees $a\geq 11$ depend on the choice of the $f$-alphabet isomorphism $\phi$. The definition of the adjoint action of the zeta generators, $[e_1,M_a] = - [e_1,P_a]$ (equation \eqref{eq:e1Ma}, above), therefore depends on the choice of isomorphism $\phi$. The generating series, $\mathbb M$, exhibits a compensating dependence on the choice of isomorphism via the MZV coefficients $\phi^{-1}(f_W)$ in its expansion (\ref{defmser}). By the joint effort of both $\phi$-dependences, the conjugation series $(\MM^{\dR})^{-1} \GG^\mot [e_0,e_1;z] \,\MM^\dR$ is independent on the choice of $f$-alphabet isomorphism, and Proposition \ref{prop:1thm} holds for any choice of isomorphism~$\phi$ satisfying the defining properties, (\ref{normf}) to (\ref{eq:fdelta}), reviewed in Section \ref{sec:2.3}.
\end{remark}

\section{The coaction formula in the multi-variable case}
\label{sec:6}

The goal of this section is to prove Lemma \ref{lem:DIcon} (equation (\ref{keyeq}) in Section \ref{sec:3.1}), which implies the new coaction formula, Theorem \ref{thm:main}, for MPLs of multiple variables (Theorem \ref{thm:coactionB}, below). Lemma \ref{lem:DIcon} rewrites the action of the zeta generators $(\MM^\dR)^{-1} e_{1,\ell} \MM^\dR$ as a conjugation of the braid generators $e_{1,\ell}$ in $\GG^\mot_1$ by the series $\GG^\dR_n\cdots\GG^\dR_2 Z^\dR_\ell$, for each of $\ell=2,3,\ldots,n{+}1$. For this purpose, we need to recall the linear combinations of braid generators, $E_0^{(r)}$ and $E_1^{(r)}$, defined in \eqref{eq:E0E1}, and derived in Lemma \ref{lem:bigphiprod} from the action of the braid group on MPLs:
\begin{align}
E_0^{(2a-1)} &= e_{1,0} + \sum_{ i=2}^a e_{1,i} \, , & E_0^{(2a)} & = e_{1,a+1} \, ,\label{eq:E0E1-repeat} \\
E_1^{(2a-1)} &= e_{1,a+1} \, , & E_1^{(2a)} &= e_{0,a+1}+\sum_{i=2}^a e_{i,a+1} \, . \notag  
\end{align}
In terms of these $E_0^{(r)}$ and $E_1^{(r)}$, Lemma \ref{lem:PhiXPhi}, above, is a formula for conjugating $E_1^{(1)}$ with the Drinfeld associator $\Phi^\dR(E_0^{(1)},E_1^{(1)})$. Lemma \ref{lem:DIcon} is a generalisation of Lemma \ref{lem:PhiXPhi} to all $n\geq 1$. 

\subsection{Products of Drinfeld associators}
\label{sec:6.1}

For any Lie polynomial $P$ in the free Lie algebra $\mathfrak{g} = {\rm Lie}\left<e_0,e_1\right>$, let us write $P^{(r)}$ for the Lie polynomial affected by the replacements $e_0 \rightarrow E_0^{(r)}$ and $e_1 \rightarrow E_1^{(r)}$. Treating the $E_0^{(r)}$ and $E_1^{(r)}$ as formal variables, they themselves define a (graded) free Lie algebra
\beq
\mathfrak{G}_n = \text{Lie}\left<\left.E_0^{(r)},E_1^{(r)}\right| 1\leq r \leq 2n{-}1\right> \, .
\eeq
By analogy with the definition of Ihara derivations in Section \ref{sec:5.1}, we can define derivations also on $\mathfrak{G}_n$. For some $r$ and a Lie polynomial $P \in \mathfrak{g}$, define the generalised Ihara derivation~as
\beq \label{eq:genihara1}
D_{P^{(r)}} E_0^{(r)} = 0 \, , \qquad D_{P^{(r)}} E_1^{(r)} = [E_1^{(r)}, P^{(r)}] \, .
\eeq
The action of $D_{P^{(r)}}$ can be extended to a Lie derivation on the whole of $\mathfrak{G}_n$. For $r<s$, and two Lie polynomials $P,Q \in \mathfrak{g}$, define 
\beq \label{eq:genihara2}
D_{P^{(r)}} Q^{(s)} = [ Q^{(s)} , P^{(r)}]\, , \qquad D_{P^{(s)}} Q^{(r)} = 0 \, .
\eeq

\begin{lemma}\label{lem:commute}
For distinct $r\neq s$, and any Lie polynomials $P,Q \in \mathfrak{g}$, the generalised Ihara derivations commute on $\mathfrak{G}_n$:
\beq
[D_{P^{(r)}}, D_{Q^{(s)}}] = 0\, .
\eeq
\end{lemma}
\begin{proof}
Fix some $q,r,s,t$ with $q<r<s<t$, and some Lie polynomials $P,Q,R \in \mathfrak{g}$. Then clearly
\beq
[D_{P^{(r)}}, D_{Q^{(s)}}] R^{(q)} = 0 \, .
\eeq
Also,
\beq
[D_{P^{(r)}}, D_{Q^{(s)}}] R^{(r)} = -D_{Q^{(s)}} \left(D_{P^{(r)}}R^{(r)}\right) =  0
\eeq
because $D_{P^{(r)}}R^{(r)}$ is a Lie polynomial in $E_0^{(r)},E_1^{(r)}$, and $s>r$. Similarly,
\beq
[D_{P^{(r)}}, D_{Q^{(s)}}] R^{(s)} = [D_{Q^{(s)}}R^{(s)},P^{(r)}] - [D_{Q^{(s)}}R^{(s)},P^{(r)}] =  0\, .
\eeq
Finally,
\beq
[D_{P^{(r)}}, D_{Q^{(s)}}] R^{(t)} = [[R^{(t)},P^{(r)}],Q^{(s)}] + [[P^{(r)},Q^{(s)}],R^{(t)}] + [[Q^{(s)},R^{(t)}],P^{(r)}] = 0\, ,
\eeq
which vanishes by the Jacobi identity. Having checked these 4 cases, we see that the Lemma follows, using the derivation property, for all Lie polynomials in $\mathfrak{G}_n$.
\end{proof}

Write $\Phi^{(r)}$ for $\Phi^\dR(E_0^{(r)}, E_1^{(r)})$ (as in Section \ref{sec:4.3}). Then we are interested in the products
\beq
\Phi^{\{ \ell \}} = 
\Phi^{(1)} \Phi^{(2)}\ldots \Phi^{(2\ell-3)}\, ,
\label{1to2ell3}
\eeq
and we find that

\begin{lemma}\label{lem:phiellX}
For $\ell \geq 1$, the conjugation by the series $\Phi^{\{ \ell \}} $ admits the following expansion in terms of MZVs
\beq
\Phi^{\{ \ell \}} E_1^{(2\ell-3)} \left(\Phi^{\{ \ell \}}\right)^{-1} = \sum_{W=ab\ldots c} \phi^{-1}(f_{\alpha(W)}^\dR) D^{ \{ \ell \} }_a D^{ \{ \ell \} }_b \ldots D^{ \{ \ell \} }_c \, E_1^{(2\ell-3)}\, ,
\eeq
where
\beq
D^{ \{ \ell \} }_a = \sum_{r=1}^{2\ell-3} D_{P_a^{(r)}}\, .
\label{defdda}
\eeq
Here, the sum is over words $W$ in the odd numbers ($3,5,7,\ldots$) and $P^{(r)}_a \in \mathfrak{g}$ are the Lie polynomials defined by the Drinfeld associator with arguments $E_0^{(r)},E_1^{(r)}$ (Lemma \ref{lem:liecoef}). Recall that $E_1^{(2\ell-3)} = e_{1,\ell}$. 
\end{lemma}

\begin{proof}
First, by the definition of the generalised Ihara derivations, \eqref{eq:genihara1}, Lemma \ref{lem:PhiXPhi} implies that
\beq\label{eq:ellproof1}
\Phi^{(L)} E_1^{(L)} \left(\Phi^{(L)}\right)^{-1} = \sum_{W=ab\ldots c} \phi^{-1}(f^\dR_{\alpha(W)}) D_{P_a^{(L)}} D_{P_b^{(L)}}\ldots D_{P_c^{(L)}} E_1^{(L)}\, ,
\eeq
where we write $L = 2\ell{-}3$. 

Second, fix some $r$. Let $X$ be some Lie series in the $E_0^{(s)}, E_1^{(s)}$ for all $s>r$. By the definition of the generalised Ihara derivations, \eqref{eq:genihara2}, note that
\beq
D_{P_a^{(r)}} X = [X,P_a^{(r)}] \, .
\eeq
Then, by the proof of Lemma \ref{lem:PhiXPhi} (i.e.\ with $e_1$ replaced by $X$, etc.), it follows that
\beq\label{eq:ellproof2}
\Phi^{(r)} X  (\Phi^{(r)} )^{-1} = \sum_{W=ab\ldots c} \phi^{-1} (f^\dR_{\alpha(W)} ) D_{P_a^{(r)}} D_{P_b^{(r)}}\ldots D_{P_c^{(r)}} X\, .
\eeq

Third, fix some $q,r$ with $q<r$, and take $X$ as above. Then by two applications of \eqref{eq:ellproof2}, and using the shuffle property of MZVs,
\begin{align}
\Phi^{(q)}\Phi^{(r)} X \left(\Phi^{(q)}\Phi^{(r)}\right)^{-1} &= \sum_{W} \phi^{-1}\left(f^\dR_{\alpha(W)} \right) \sum_{\substack{A = a\ldots b \\ B = c\ldots d}}(W,A\shuffle B)  \\
&\quad 
\times D_{P_a^{(q)}}\ldots D_{P_b^{(q)}} D_{P_c^{(r)}}\ldots D_{P_d^{(r)}} X \, .\notag
\end{align}
Note that in the sum over shuffles, every letter of $W$ appears precisely once in either $A$ or~$B$. Here we can apply Lemma \ref{lem:commute}: the $D_{P_a^{(q)}}$ do not commute among themselves, but we are free to commute each of the $D_{P_c^{(r)}}$ with each of the $D_{P_a^{(q)}}$. It follows that
\beq\label{eq:ellproof3}
\Phi^{(q)}\Phi^{(r)} X \left(\Phi^{(q)}\Phi^{(r)}\right)^{-1} = \! \sum_{W=a\ldots c}  \! \phi^{-1} (f^\dR_{\alpha(W)}) \left( D_{P_a^{(q)}} {+} D_{P_a^{(r)}} \right) \ldots \left( D_{P_c^{(q)}} {+} D_{P_c^{(r)}} \right) X\, .
\eeq
Combined with \eqref{eq:ellproof1}, repeated applications of \eqref{eq:ellproof3} gives the Lemma.
\end{proof}

\subsection{Proof of Theorem \ref{thm:main}} 
\label{sec:6.2}
The proof for our coaction formula, Theorem \ref{thm:main}, now follows from the foregoing Lemmas in Section \ref{sec:6.1} and Section \ref{sec:5.2}. First, we can define an adjoint action of the zeta generators $M_a$ (for $a = 3,5,7,\ldots$) by
\begin{align}
[E_0^{(\ell)}, M_a]  = 0\, , \qquad
[E_1^{(\ell)}, M_a]  = -\left[E_1^{(\ell)}, \sum_{r=1}^{2\ell-3} P_a^{(r)}\right] \, ,
\label{eq:adaction} 
\end{align}
where $E_0^{(1)} = e_{1,0}$ and $E_1^{(2\ell-3)} = e_{1,\ell}$. This action of zeta generators can also be obtained by specializing the more general results of \cite{GTLeila} to specific choices of point numbering, pants decomposition and elements of the Grothendieck-Teichm\"uller group in the terminology of the reference. In our case, the choice in (\ref{eq:adaction}) is inspired by Lemma \ref{lem:phiellX}, and implies the following identity:

\begin{lemma}\label{lem:proof3}
With the adjoint action defined above and $\Phi^{\{ \ell \}}$ defined by (\ref{1to2ell3}),
\begin{equation}
\Phi^{\{ \ell \}} E_1^{(2\ell-3)} \left(\Phi^{\{ \ell \}}\right)^{-1} = \left(\MM^\dR\right)^{-1} E_1^{(2\ell-3)} \MM^\dR\, .
\end{equation}
\end{lemma}
\begin{proof}
Lemma \ref{lem:phiellX} gives
\beq\label{eq:proof31}
\Phi^{\{ \ell \}} E_1^{(2\ell-3)} \left(\Phi^{\{ \ell \}}\right)^{-1} = \sum_{A=ab\ldots c} \phi^{-1}(f^\dR_{\alpha(A)}) D^{ \{ \ell \} }_aD^{ \{ \ell \} }_b \ldots D^{ \{ \ell \} }_c E_1^{(2\ell-3)} \, ,
\eeq
where $D_a^{\{ \ell \}}$ is defined by (\ref{defdda}). On the other hand, Lemma \ref{lem:MXM} implies that
\beq\label{eq:proof32}
(\MM^\dR)^{-1} E_1^{(2\ell-3)} \MM^\dR = \sum_{A=ab\ldots c} \phi^{-1}(f^\dR_A) \, \ell[E_1^{(2\ell-3)},M_a,M_b,\ldots,M_c] \, .
\eeq
However, by the definition of the generalised Ihara derivations,
\beq
D_a^{\{ \ell \}} E_1^{(2\ell-3)} = \left[E_1^{(2\ell-3)}, \sum_{r=1}^{2\ell-3} P_a^{(r)}\right] = - [E_1^{(2\ell-3)}, M_a] \, .
\label{adnvar}
\eeq
Moreover, ${\rm ad}_{M_a}$ is itself a Lie derivation on $\mathfrak{G}_n$. So, comparing \eqref{eq:proof32} with \eqref{eq:proof31}, the Lemma follows.
\end{proof}

Recall from Section \ref{sec:2.4} that the Ihara formula for the motivic coaction reads
\begin{equation}\label{eq:Iharaagain}
\Delta \GG^\mot_1[\{e_{1,i}\};z_1] = \GG^\mot_1[\{e'_{1,i}\};z_1] \GG^\dR_1[\{e_{1,i}\};z_1] \, ,
\end{equation}
where $\GG^\mot_1$ features a change of alphabet  from $e_{1,\ell}$ to $e'_{1,\ell}$ (see (\ref{ykconj}) and (\ref{zkconj})) for $\ell=2,3,\ldots,n{+}1$, whereas $e'_{1,0} = e_{1,0}$. However, we can now show that

\begin{lemma}\label{lem:DIcon}
The change of alphabet to
\begin{equation}
e'_{1,\ell} = 
Z^\dR_\ell
\, e_{1,\ell}\,
(Z^\dR_\ell)^{-1}\, , \qquad
Z^\dR_\ell = \lim_{z_1\rightarrow z_\ell} \GG_1^\dR[ \{e_{1,i} \} ;z_1]
\label{chalphb}
\end{equation}
for $\ell=2,3,\ldots,n{+}1$ can equivalently be written as
\begin{equation}
e'_{1,\ell} = (\MM^\dR \GG_n^\dR \cdots \GG_2^\dR)^{-1}\, e_{1,\ell}\, (\MM^\dR \GG_n^\dR \cdots \GG_2^\dR)\, .
\end{equation}
\end{lemma}

\begin{proof}
The Lemma is equivalent to showing that
\begin{equation}
(\MM^\dR)^{-1} e_{1,\ell} \MM^\dR = (\GG^\dR_n\cdots\GG^\dR_2 Z^\dR_\ell)\, e_{1,\ell}\, (\GG^\dR_n\cdots\GG^\dR_2 Z^\dR_\ell)^{-1}\, .
\end{equation}
By Lemmas \ref{lem:YeY} and \ref{lem:bigphiprod}, the right-hand side is equal to $\Phi^{\{ \ell \}} E_1^{(2\ell-3)} \left(\Phi^{\{ \ell \}}\right)^{-1}$, and we also recall that $E_1^{(2\ell-3)} = e_{1,\ell}$ by \eqref{eq:E0E1}. But then the result follows from Lemma \ref{lem:proof3}, above.
\end{proof}

\begin{remark} 
In Lemma \ref{lem:DIcon}, all three ingredients of our results, highlighted in Section \ref{sec:3.3}, come together: (i) the $z_i$-independence of $(\mathbb G_n^\dR\cdots \mathbb G_2^\dR Z^\dR_\ell) e_{1,\ell} (\mathbb G_n^\dR\cdots \mathbb G_2^\dR Z^\dR_\ell)^{-1}$, (ii) the use of the braid group action to rewriting this conjugation using Drinfeld associators, and (iii) the replacement of these Drinfeld associators by the group-like series $\mathbb M^\dR$ in zeta generators.
\end{remark}

\begin{theorem}\label{thm:coactionB}
With the adjoint action defined above, the motivic coaction acts on the generating series $\GG_1^\mot$ as
\begin{equation}\label{eq:thmB}
\Delta \GG_1^\mot = (\HH^\dR_n)^{-1}\, \GG_1^\mot \, \HH^\dR_n \, \GG_1^\dR\, ,
\end{equation}
where
\begin{equation}
\HH_{n} = \MM \, \GG_n  \cdots \GG_{2}  \, .
\label{exphn}
\end{equation}
\end{theorem}

\begin{proof}
Lemma \ref{lem:DIcon}, above, shows that the change of alphabet (\ref{chalphb}) is given by
\begin{equation}
e'_{1,\ell} = (\HH^\dR_n)^{-1} e_{1,\ell} \HH^\dR_n\, .
\end{equation}
Clearly for two $i,j> 1$,
\beq\label{eq:coactionB2}
e'_{1,i}e'_{1,j} = (\HH^\dR_n)^{-1} e_{1,i}e_{1,j} \, \HH^\dR_n\, .
\eeq
Moreover, $e_{1,0}$ (which coincides with $e_{1,0}^\ast$ in (\ref{adaptser})) commutes with $\GG_n^\dR \cdots \GG_{2}^\dR$, by \eqref{eq:easteast1}. Also, $e_{1,0}  = E_0^{(1)}$ (equation \eqref{eq:E0E1-repeat}), so by the definition of the adjoint action, \eqref{eq:adaction}, $e_{1,0} $ commutes with $\MM^\dR$. We conclude that
\beq
[e_{1,0},\HH^\dR_n] = 0\, ,
\eeq
and, together with \eqref{eq:coactionB2}, it follows that the conjugation by the series $\HH^\dR_n$ can be `pulled out' of the MPL generating series:
\begin{equation}
\GG^\mot_1[\{e'_{1,i}\};z_1] 
= (\HH^\dR_n)^{-1}  \,
\GG^\mot_1[\{e_{1,i}\};z_1] 
\, \HH^\dR_n\, .
\end{equation}
The Theorem then follows from the Ihara formula, \eqref{eq:Iharaagain}.
\end{proof}

\begin{remark}
We again emphasize that the conjugation by the series $\mathbb M^\dR$ in the main formula, equation \eqref{eq:thmB}, does \emph{not} depend on the choice of the $f$-alphabet isomorphism $\phi$ that defines the MZVs $\phi^{-1}(f^\dR_W)$. This is because the choice of $\phi$ also changes the Lie polynomials $P_a^{(r)}$ that appear in the definition, equation \eqref{eq:adaction}, of the action of the zeta generators $M_a$. 
\end{remark}

Finally, we note that in order to expand the coaction formula, \eqref{eq:thmB}, given by our Theorem, it is necessary to expand the inverse series $(\MM^\dR)^{-1}$ and $(\GG_k^\dR)^{-1}$ (for $2\leq k \leq n$). The inverse series $(\MM^\dR)^{-1}$ is given by Lemma \ref{lem:Minv}. A similar argument gives
\begin{lemma}\label{lem:GMinverse}
The series inverse of $\GG^\dR({\cal A};z)$, for some alphabet ${\cal A} = \{a_1,\ldots,a_n\}$, is
\beq
\left( \GG^\dR \left[\begin{matrix} e_{a_1} & \cdots & e_{a_n} \\ a_1 & \cdots & a_n \end{matrix} ; z \right] \right)^{-1} = \sum_W e_{W^t}\,G^\dR (\alpha(W); z) = \sum_W (-1)^{|W|} e_{W}\,G^\dR (W; z)\, ,
\label{invsgg}
\eeq
where the sum is over words $W \in {\cal A}^\times$, including the empty word, and $W^t$ denotes the reversed word.

\end{lemma}
\begin{proof}
Using the shuffle product $G(A;z) G(B;z) = G(A\shuffle B;z)$ of MPLs (equation (\ref{eq:Gshuffle})), the following product
involving the series in \eqref{invsgg} simplifies to
\beq
\GG({\cal A};z) \, 
\sum_W e_{W^t}\,G(\alpha(W); z) = \sum_C e_{C^t} \sum_{C = BA} G(A \shuffle \alpha(B) ;z)\, ,
\eeq
where $\alpha(W) = (-1)^{|W|} W^t$ is the antipode for the free Hopf algebra (\eqref{lie:antipode1} in Appendix \ref{app:Lie}), and the sum is over $C \in {\cal A}^\times$. However, for all non-empty words $C$,
\beq
\sum_{C = BA} G(A \shuffle \alpha(B) ;z) = 0
\eeq
by the antipode identity for the free Hopf algebra (equation \eqref{lie:antipode2} in Appendix \ref{app:Lie}).
\end{proof}

\section{The single-valued map}
\label{sec:7}

As discussed in Section \ref{sec:mdr}, the algebra of de Rham MPLs, ${\cal P}^\dR$, is a Hopf algebra which is equipped with the shuffle product (\ref{eq:Gshuffle}), to be denoted by $\mu$ in this section, a coproduct, and an antipode, $S$, to be introduced in Section \ref{sec:7.1}, below. The antipode can be used to construct the single-valued map of MPLs using the motivic coaction, $\Delta$, as reviewed in Section \ref{sec:7.2}. We use our coaction formula, Theorem \ref{thm:main}, to find formulas for the action of the antipode, $S$, and prove our formula for the single-valued map, Theorem \ref{thm:sv}, in Section \ref{sec:7.proof}. An alternative proof by direct matching with the construction of single-valued MPLs in \cite{DelDuca:2016lad} can be found in Section~\ref{sec:7.3}.

The algebra of motivic MPLs, ${\cal P}^\mot$, is a Hopf algebra comodule (see Section \ref{sec:mdr}). We can define a projection $\Pi^\dR : {\cal P}^\mot \rightarrow {\cal P}^\dR$, by replacing motivic MPLs, $G^\mot$, with the corresponding de Rham MPLs, $G^\dR$,
\beq
\Pi^\dR\big( G^\mot(A;z) \big) =  G^\dR(A;z) \, ,
\label{drproj}
\eeq
for some word $A \in {\cal A}^\times$ in the alphabet ${\cal A}$, see Section 4.3 of \cite{Francislecture}. In particular, $\Pi^\dR((2\pi i)^\mot)=0$, and the de Rham projection removes terms of $G^\mot$ that correspond to monodromies, i.e. terms with at least one power of $(2\pi i)^\mot$.

The coaction $\Delta Z^\mot$, for some motivic MPL $Z^\mot$, is a sum of terms of the form $X^\mot Y^\dR$, involving both motivic and de Rham, which cannot be multiplied. To define multiplication we again use the projection $\Pi^\dR$.

\begin{definition}\label{rem:not}
Write $\mu$ for multiplication on ${\cal P}^\dR$. Then
\beq
\mu \circ (\Pi^\dR \otimes 1) \circ \Delta
\eeq
defines a homomorphism from ${\cal P}^\mot$ to ${\cal P}^\dR$. Here $\mu \circ (\Pi^\dR \otimes 1)$ acts on some term $X^\mot Y^\dR$ as
\beq
\mu \circ (\Pi^\dR \otimes 1): X^\mot Y^\dR \mapsto X^\dR \cdot Y^\dR.
\eeq
\end{definition}

\subsection{The antipode of $\MM$ and $\GG_k$}
\label{sec:7.1}
The antipode, $S$, for the Hopf algebra of de Rham MPLs is determined by the antipode identity \cite{Goncharov:2001iea}
\beq\label{sv:antipode0}
\mu \circ (S \circ \Pi^\dR \otimes 1) \circ \Delta G^\mot(A;z) = 0 \, , \ \ \ \ A \neq \emptyset\, ,
\eeq
where $S$ acts just on the de Rham projection (\ref{drproj}) of the motivic entry in 
$\Delta G^\mot(A;z)$. (See \ref{rem:not} for more comments on the notation used here.) In terms of our generating series, $\GG_1^\mot$ (Section \ref{sec:2.2}), the antipode identity reads
\beq\label{sv:antipode1}
\mu \circ (S \circ \Pi^\dR \otimes 1)  \circ \Delta \GG^\mot_1 = 1 
\eeq
since the leading term, $1$, of the generating series, $\GG^\mot_1= 1+\ldots$, is left unchanged by the map: $\mu \circ (S \circ \Pi^\dR \otimes 1)  \circ \Delta 1=1$.

\begin{lemma}\label{lem:SGHGH}
The antipode acts on the de Rham version of the generating series $\GG_1$ as
\beq\label{sv:SGHGH}
S \GG^\dR_1 = \HH^\dR_n (\GG^\dR_1)^{-1} (\HH^\dR_n)^{-1} \, ,
\eeq
where the product $\HH_{n}$ of generating series in MPLs and MZVs is defined by (\ref{exphn}).
\end{lemma}
\begin{proof}
Recall that, by Theorem \ref{thm:main},
\beq
\Delta \GG^\mot_1 = \left( \HH_n^\dR\right)^{-1}   \GG_1^\mot  \HH^\dR_n \GG_1^\dR.
\eeq
So
\beq
\mu \circ (S \circ \Pi^\dR \otimes 1) \circ \Delta \GG^\mot_1 = \left( \HH_n^\dR\right)^{-1} \left( S \GG_1^\dR \right) \HH^\dR_n \GG_1^\dR.
\eeq
Then, by the antipode identity, \eqref{sv:antipode1},
\beq
1 = \left( \HH^\dR_n \right)^{-1} \left( S \GG^\dR_1 \right) \HH^\dR_n \GG^\dR_1 \, ,
\eeq
which implies the Lemma by multiplying through by the appropriate series inverses.
\end{proof}

\begin{remark}
Note that, in the case of $n=1$, we have $\HH_{1} = \MM$, and (\ref{sv:SGHGH}) reduces to
\beq
S\GG^\dR[e_0,e_1;z]   =\MM^\dR \GG^\dR[e_0,e_1;z]^{-1} (\MM^\dR)^{-1}\, .
\label{S1var}
\eeq
Expanding the generating series, we see that, for some word $A \in \{0,1\}^\times$,
\beq
S G^\dR(A;z)=  G^\dR\big(\alpha(A);z\big)+\ldots,
\eeq
where the ellipsis comprises products of one-variable de Rham MPLs $G^\dR(B;z)$ of weight $|B|\leq |A|{-}3$, multiplied by de Rham MZVs $\phi^{-1}(f^\dR_W)$, with $W \in \{ 3,5,\ldots \}^\times$. Here, we have used Lemma \ref{lem:GMinverse} for the series expansion of $(\GG^\dR)^{-1}$, and $\alpha$ is the antipode on the free Hopf algebra of words (see (\ref{eq:antipodeA}) or Appendix~\ref{app:Lie})). Similarly, again using Lemma \ref{lem:GMinverse}, expanding equation (\ref{sv:SGHGH}) for $n>1$ with more general letters $A\in \{ 0,1,z_2,\ldots,z_n\}^\times$ gives
\beq
S G^\dR(A;z_1)=  G^\dR\big(\alpha(A);z_1\big) + \ldots
\, ,
\eeq
where the ellipsis gather the corrections due to the conjugation by $\HH^\dR_n=1+\ldots$ in (\ref{sv:SGHGH}). These terms take the form of a $z_1$-dependent polylogarithm, $G^\dR(B;z_1)$, at weight $|B|\leq |A|-1$, multiplied by de Rham MPLs in $\leq n{-}1$ variables and de Rham MZVs. Hence, the conjugation by the series $\HH^\dR_n$ in (\ref{sv:SGHGH}) interpolates between the antipode $\alpha$ on words and the antipode $S$ on MPLs. 
\end{remark}

We also find formulas for the action of $S$ on de Rham MZVs. 
In terms of the $f$-alphabet, the motivic coaction $\Delta$ on MZVs was given in Section \ref{sec:2.3} (equations \eqref{eq:fcoaction} and \eqref{eq:fdelta}). Also, the multiplication $\mu$ corresponds to the shuffle product in the $f$-alphabet representation (equation  \eqref{eq:fshuffle}). The antipode identity for a de Rham MZV, $\phi^{-1}(f_W^\dR)$, is then
\beq\label{sv:motivicS}
\mu \circ (S \circ \Pi^\dR \otimes 1)  \circ \Delta\, \phi^{-1}\big( f_W^\mot\big) = 0\, ,
\eeq
where the antipode $S$ acts only on the de Rham projection of the motivic entry of $ \Delta\, \phi^{-1}\big( f_W^\mot\big)$. 

\begin{lemma}\label{lem:SM}
The antipode acts on $\MM^\dR$, the generating series of de Rham MZVs, as
\beq\label{eq:SM}
S\MM^\dR = (\MM^\dR)^{-1} \, .
\eeq
\end{lemma}
\begin{proof}
Recall that  $\Delta \MM^\mot= \MM^\mot \MM^\dR$ (equation \eqref{eq:Mcoaction} in Section \ref{sec:2.3}). So the antipode identity, \eqref{sv:motivicS}, becomes
\begin{align}
1 &=\mu \circ (S \circ \Pi^\dR \otimes 1)  \circ\Delta \MM^\mot \notag \\
&= \mu \circ (S \circ \Pi^\dR \otimes 1) \MM^\mot \MM^\dR
\notag \\
&= (S \MM^\dR) \MM^\dR \, ,
\end{align}
which implies the Lemma.
\end{proof}

\begin{remark}
Using the expansion of the series $\MM^{-1}$, Lemma \ref{lem:Minv}, we see that this Lemma implies that $S$ acts on de Rham MZVs as
\beq
S \phi^{-1}(f^\dR_W)= \phi^{-1}(f^\dR_{\alpha(W)}) \,,
\eeq
for some word $W \in \{3,5,7,\ldots\}^\times$. Or, more explicitly, for $i_1,i_2,\ldots,i_r \in 2\mathbb N{+}1$,
\beq
S \phi^{-1}(f_{i_1} f_{i_2}\ldots f_{i_r} )^\dR =(-1)^r \phi^{-1}(f_{i_r} \ldots f_{i_2} f_{i_1})^\dR \, .
\eeq
using the definition of the antipode $\alpha$ on the free Hopf algebra (see equation (\ref{eq:antipodeA}) or Appendix~\ref{app:Lie}).
\end{remark}

\subsection{The sv map and the antipode}
\label{sec:7.2}
Single-valued MPLs are closely related to the antipode $S$ and the coaction $\Delta$. To explain this connection, we will introduce the \emph{signed complex conjugate} of an MPL,
\beq
\widetilde{G^\dR}(A;z) \coloneqq (-1)^{|A|} \, \overline{G^\dR(A;z)} \, ,
\label{tildop.01}
\eeq
where $\overline{G^\dR(A;z)}$ is the complex conjugate, and $|A|$ is the length of $A$ (i.e.\ the weight of $G^\dR(A;z)$). Write $\widetilde{\GG_k^\dR}$ for the series obtained from $\GG_k^\dR$ by applying $\,\widetilde{\cdot}\,$ to the MPLs in the series. Then, using again Lemma \ref{lem:GMinverse}, we see that, 
\beq\label{sv:tildeGinv}
\widetilde{\GG_k^{-1}} = \overline{\GG_k^t}\, .
\eeq
Moreover, taking special values of de Rham MPLs, we obtain the signed conjugates of MZVs,
\beq
\widetilde{\phi^{-1}(f_W^\dR)} = (-1)^{|W|} \, \phi^{-1}(f_W^\dR)\, ,
\label{tildop.03}
\eeq
for any de Rham MZV $\phi^{-1}(f^\dR_W)$, where we use that de Rham MZVs are real. Equivalently, in terms of the generating series, we can write
\beq\label{sv:tildeMinv}
\widetilde{(\MM^\dR)^{-1}} = (\MM^\dR)^t \, ,
\eeq
where we again use the expansion of the series $(\MM^\dR)^{-1}$ (Lemma \ref{lem:Minv}).

Let $\widetilde{S}$ denote the action of the antipode, $S$, followed by $\,\widetilde{\cdot}\,$. Then the \emph{single-valued map} on motivic MPLs admits the combinatorial construction \cite{Brown:2013gia, DelDuca:2016lad} (also see \cite{Charlton:2021uhu})
\beq\label{sv:sv}
\sv = \mu \circ 
(\widetilde S \circ \Pi^\dR \otimes 1)
\circ \Delta\, ,
\eeq
where $\widetilde{S}$ acts only on the de Rham projection of the motivic part coming from the coaction~$\Delta$. As written, this is a map from ${\cal P}^\mot$ to ${\cal P}^\dR$. We abuse notation and suppress a final map to (i) convert the de Rham period obtained from (\ref{sv:sv}) into a motivic period and to (ii) evaluate the latter to a complex function via the period map (see Remark \ref{rem:per}). The complex function, $\sv\, G^\mot(A;z)$, obtained by applying \eqref{sv:sv} to a motivic MPL, is single-valued in all its variables (i.e.\ $z$ and also the non-constant letters in $A$) and satisfies the same holomorphic differential equation as $G(A;z)$ (equation (\ref{coact.09})).

\begin{remark}\label{rem:per}
The period map, ${\rm per}$, evaluates a motivic period to obtain a complex function (or complex number), and we write the result by dropping superscripts: ${\rm per} \, G^\mot(A;z) = G(A;z)$ and ${\rm per} \, \zeta^\mot_{n_1,\ldots,n_r} = \zeta_{n_1,\ldots,n_r}$. We suppress the period map when writing ${\rm sv}$ in \eqref{sv:sv}. Note that, to apply ${\rm per}$ in equation \eqref{sv:sv}, we must also choose an algebra homomorphism from ${\cal P}^\dR$ to ${\cal P}^\mot$. However, it is known that the result does not depend on this choice \cite{Brown:2013gia} (and see \textit{Remark 3.6}).
\end{remark}

\begin{remark}
\label{svdrm}
The ${\rm sv}$ map we define here, in \eqref{sv:sv} differs slightly from the single-valued maps defined in  \cite{Brown:2013gia, Francislecture, Brown:2018omk}. These references define a map from ${\cal P}^\dR$ to ${\cal P}^\mot$ (which can then be composed with the period map, whereas \eqref{sv:sv} defines a map from ${\cal P}^\mot$ to ${\cal P}^\dR$. However, the map in these references can be written in a form that is very similar to \eqref{sv:sv}: it is given by $\mu \circ (\widetilde{S} \otimes 1) \circ \Delta'$, where $\Delta'$ is the coproduct $\Delta': {\cal P}^\dR \rightarrow {\cal P}^\dR \otimes {\cal P}^\dR$. The coproduct $\Delta'$ can be defined by formulas almost identical to that for the coaction in (\ref{GBcoact}), except with replacing the $\mot$ superscripts on both sides by $\dR$. Indeed, the coproduct on ${\cal P}^\dR$ is related to the motivic coaction by
\beq
\Delta' \circ \Pi^\dR = (\Pi^\dR \otimes 1) \circ \Delta\, ,
\eeq
so that our definition \eqref{sv:sv} is equivalent to the earlier one.
\end{remark}

\begin{remark}\label{rem:mult}
It is known that the single-valued map (\ref{sv:sv}) is a homomorphism, and in particular, that it is multiplicative: 
\beq
\sv\,(X^\mot Y^\mot) = (\sv\, X^\mot) \,(\sv\,Y^\mot)\,,
\eeq
for arbitrary motivic MPLs $X^\mot$ and $Y^\mot$. Indeed, this follows from the multiplicative property of the motivic coaction, $\Delta$,
\beq\label{sv:DeltaXY}
\Delta (X^\mot Y^\mot) = \Delta (X^\mot) \Delta (Y^\mot)\, ,
\eeq
where, on the right-hand side, we take products separately in ${\cal P}^\dR$ and ${\cal P}^\mot$. This property, \eqref{sv:DeltaXY}, can be taken as a defining property of a coaction. Moreover, the antipode $S$, and $\widetilde{S}$, all respect multiplication in ${\cal P}^\dR$. A more detailed account on the multiplicativity of the single-valued map can for instance be found in appendix B of \cite{DelDuca:2016lad}.
\end{remark}

\subsection{Proof of the sv map formula}\label{sec:7.proof}
Applied to our generating series $\GG_1$ of MPLs, the ${\rm sv}$ map, equation \eqref{sv:sv} gives a formula for a generating series of single-valued MPLs: $\sv\,\GG^\mot_1$. This leads to our formula for the single-valued map (Theorem \ref{thm:sv}), given as Theorem \ref{thm:sv7}, below. First, note that the single-valued map, equation \eqref{sv:sv}, can also be applied to motivic MZVs, and to generating series of MZVs.\footnote{Alternatively, single-valued MZVs could be obtained by taking (regularized) $z\rightarrow 1$ limits of single-valued MPLs. It is surprising that the single-valued map is consistent with all the relations among motivic MZVs. However, it follows from Theorem 1.1 and Corollary 5.4 of \cite{Brown:2013gia} that the single-valued map of motivic MZVs is well-defined and commutes with evaluation of MPLs.} We find:

\begin{lemma}\label{lem:svMMM}
A generating series $\sv\,\MM^\mot$ of single-valued MZVs is given by
\beq\label{sv:svMMM}
\sv\, \MM^\mot= \MM^t \MM \, .
\eeq
\end{lemma}
\begin{proof}
The coaction $\Delta \MM^\mot = \MM^\mot \MM^\dR$ in Section \ref{sec:2.3} implies that
\beq
\sv \,\MM^\mot =  \mu \circ (\widetilde S \circ \Pi^\dR \otimes 1)\circ \Delta \MM^\mot 
=  \mu \Big( \big(  \widetilde{S}\MM^\dR \big) \otimes \MM^\dR \Big) \, .
\eeq
Then, by Lemma \ref{lem:SM} and equation \eqref{sv:tildeMinv}, $\widetilde{S} \MM^\dR = (\MM^\dR)^t$, so that
\beq
\sv \,\MM^\mot =   \mu \big( (\MM^\dR)^t\otimes  \MM^\dR \big) = \MM^t \MM\, .
\eeq
\end{proof}

\begin{remark}
This Lemma agrees with the $f$-alphabet representation \cite{Schnetz:2013hqa, Brown:2013gia}
\beq 
{\rm sv} \, f^\mot_W = \sum_{W=AB} f_{A^t} \shuffle f_B
\, , \qquad
{\rm sv} \, (f_{i_1} f_{i_2}\ldots f_{i_r})^\mot = \sum_{j=0}^r f_{i_j}\ldots f_{i_2} f_{i_1} \shuffle f_{i_{j+1}} \ldots f_{i_r}
\eeq
of the single-valued map given in \cite{Schnetz:2013hqa, Brown:2013gia}, where $i_1,i_2,\ldots,i_r \in 2\mathbb N{+}1$. By a slight abuse of notation, we do not distinguish between the above sv-map and its composition $\phi \circ {\rm sv} \circ \phi^{-1}$ with the $f$-alphabet isomorphism.
\end{remark}

\begin{theorem}\label{thm:sv7}
The single-valued map acts on the generating series $\GG_1$ as
\begin{equation}\label{thm:sv:eq}
\sv\, \GG^\mot_1 = \left(\sv \,\HH^\mot_n\right)^{-1} \,\overline{\GG_1^t}\, \left(\sv\, \HH^\mot_n\right) \GG_1\, ,
\end{equation}
where $\HH_{n} = \MM \, \GG_n \cdots \GG_{2}$, and its single-valued version is given by
\beq\label{thm:sv:claim}
\sv\, \HH^\mot_n =  \overline{\HH}_n^{\,t} \, \HH_n \, .
\eeq
\end{theorem}

\begin{proof}
To compute the series
\beq
\sv\, \GG^\mot_1 =   \mu \circ (\widetilde S \circ \Pi^\dR \otimes 1) \circ \Delta \GG^\mot_1
\eeq
we use our main Theorem \ref{thm:main} for the coaction, which gives
\beq
\mu \circ (\widetilde S \circ \Pi^\dR \otimes 1) \circ \Delta  \GG^\mot_1 = \left( \HH^\dR_n\right)^{-1} \big( \widetilde{S}\, \GG_1^\dR \big)\, \HH^\dR_n\, \GG_1^\dR\, .
\eeq
The next step is to apply Lemma \ref{lem:SGHGH} for the antipode followed by \eqref{sv:tildeGinv} and \eqref{sv:tildeMinv},
\beq
\widetilde{S} \GG^\dR_1 = \widetilde{\HH^\dR_n}\, \widetilde{(\GG^{\dR}_1)^{-1}}\, \widetilde{(\HH_n^\dR)^{-1}}
= \overline{((\HH^\dR_n)^{-1}) ^t}\, \overline{(\GG^{\dR}_1)^{t}}\, 
\overline{(\HH_n^\dR)^{t}} \, ,
\eeq
such that, 
\beq\label{thm:sv:1}
\sv\, \GG^\mot_1 = \big( \overline{\HH}_n^{\,t} \HH_n \big)^{-1}\, \overline{\GG^{t}_1}\, \big( \overline{\HH}_n^{\,t} \HH_n \big) \GG_1\, .
\eeq
So the first part of the Theorem, \eqref{thm:sv:eq}, follows if we can prove \eqref{thm:sv:claim}. In fact, both \eqref{thm:sv:claim} and \eqref{thm:sv:eq}, now follow by a joint induction.

To see this induction, we introduce the shorthand
\beq
\HH_{n,k} = \MM \, \GG_n  \cdots \GG_{k+1}  \, , \qquad 1\leq k \leq n 
\label{defhnk}
\eeq
for a reduced variant of the product $\HH_{n}$ that omits $\GG_{2} ,\GG_{3} ,\ldots,\GG_{k} $, where, in particular, $\HH_{n,1} = \HH_n$ and $\HH_{n,n} = \MM$. Note that, after replacing $n$ by $n{-}k{+}1$ and relabelling the subscripts, \eqref{thm:sv:1} is equivalent to
\beq\label{thm:sv:2}
\sv\, \GG^\mot_k = \big( \overline{\HH}^{\,t}_{n,k} \,\HH_{n,k} \big)^{-1}\, \overline{\GG^{t}_k}\, \big( \overline{\HH}^{\,t}_{n,k} \,\HH_{n,k} \big) \GG_k\, .
\eeq
We claim that, with $n$ fixed and $1\leq k \leq n$,
\beq\label{sv:ind1}
\sv\, \GG^\mot_k = \left(\sv\, \HH^\mot_{n,k}\right)^{-1} \,\overline{\GG_k^t}\, \left(\sv\, \HH^\mot_{n,k}\right) \GG_k
\eeq
and
\beq\label{sv:ind2}
\sv\, \HH^\mot_{n,k} =\overline{\HH}^{\,t}_{n,k} \,\HH_{n,k}\, .
\eeq
These recover the Theorem for $k=1$. For $k=n$, we have $\HH_{n,n} = \MM$, and \eqref{sv:ind2} follows from $\sv\, \MM^\mot = \MM^t \MM$, see Lemma \ref{lem:svMMM}. Moreover, for $k=n$, \eqref{sv:ind1} then follows from \eqref{thm:sv:2}. Now suppose that \eqref{sv:ind1} and \eqref{sv:ind2} have been shown for all $k\geq \ell {+} 1$. Then consider
\beq
\HH_{n,\ell} = \HH_{n,\ell+1} \GG_{\ell+1}\, .
\eeq
Using multiplicativity of the single-valued map and \eqref{sv:ind1}, with $k=\ell{+}1$, we find
\beq
\sv\, \HH^\mot_{n,\ell} = \sv\, \HH^\mot_{n,\ell+1} \,\,\sv\, \GG^\mot_{\ell+1} = \overline{\GG}_{\ell+1}^{\,t}\, \left(\sv\, \HH^\mot_{n,\ell+1}\right) \GG_{\ell+1}\, .
\eeq
By \eqref{sv:ind2}, with $k=\ell{+}1$, this implies that
\beq
\sv\, \HH^\mot_{n,\ell} = \overline{\HH}^{\,t}_{n,\ell} \,\HH_{n,\ell} \, .
\eeq
So \eqref{sv:ind2} holds for $k=\ell$. Moreover, this implies \eqref{sv:ind1} for $k=\ell$ using equation \eqref{thm:sv:2}. In particular, the Theorem (case $k=1$) follows.
\end{proof}

\subsection{Alternative proof via change of alphabet}
\label{sec:7.3}

We shall here present an alternative proof of Theorem \ref{thm:sv} by directly matching the expression (\ref{thm2.2eq}) for the generating series of single-valued MPLs with their construction in \cite{DelDuca:2016lad}. The result of the reference on the generating series of single-valued MPLs in any number of variables is equivalent to
\beq
\sv\, \GG^\mot_1[e_{1,0} ,\{e_{1,\ell}\};z_1] =   \overline{\GG^t_1}\big[e_{1,0}  ,\{ \widehat e_{1,\ell} \} ; z_1 \big] \, \GG_1[e_{1,0} ,\{e_{1,\ell}\};z_1]
\label{svalt.01}
\eeq
and involves a change of alphabet for the
braid generators $e_{1,\ell}$ with $\ell = 2,\ldots,n{+}1$ similar to the multivariate Ihara formula (\ref{ykconj})\footnote{This is equivalent to equation (3.60) of \cite{DelDuca:2016lad}, where the generating series in the reference are reversed in comparison to ours (that is why we do not have $(\sv \, Z_\ell)^{-1}e_{1,\ell}  (\sv \, Z_\ell) $).},
\beq
\widehat e_{1,\ell} = (\sv \, Z^\mot_\ell) \,e_{1,\ell} \, (\sv \, Z^\mot_\ell)^{-1} \, .
\label{svalt.02}
\eeq
The motivic associators $Z^\mot_\ell$ are defined by (\ref{zkconj}) as shuffle-regularized limits $z_1 \rightarrow z_\ell$ of $ \GG^\mot_1$, and the proof of (\ref{svalt.01}) in \cite{DelDuca:2016lad} relies on the fact that the single-valued map commutes with shuffle regularization. We emphasize that the complex conjugation and the reversal prescription $^t$ of the series $\overline{\GG^t_1}[e_{1,0}  ,\{ \widehat e_{1,\ell} \} ; z_1 ]$ on the right-hand side of (\ref{svalt.01}) does not apply to the single-valued MPLs and the Lie words in braid generators obtained from the expansion of the associators in (\ref{svalt.02}). Instead, the expansion is
\beq
\overline{\GG^t_1}\big[e_{1,0}  ,\{ \widehat e_{1,\ell} \} ; z_1 \big] = 1 + \sum_{r=1}^\infty 
\sum_{i_1,\ldots,i_r=0 \atop {i_1,\ldots,i_r \neq 1 }}^{n+1}  \widehat e_{1,i_1} \ldots \widehat e_{1,i_r}
\overline{G(z_{i_1},\ldots,z_{i_r};z_1)}
\label{svalt.00}
\eeq
according to (\ref{defgg})
with $\widehat e_{1,0} = e_{1,0}$ and all the $\widehat e_{1,\ell}$ with $\ell \geq 2$ as in (\ref{svalt.02}). 

With the expansion (\ref{svalt.00}) in mind, we can rewrite the statement of Theorem \ref{thm:sv} in the alternative form
\beq
\sv\, \GG^\mot_1[e_{1,0} ,\{e_{1,\ell}\};z_1] =   \overline{ \GG_1}^t\big[e_{1,0}  \, ,\,\{(\sv \,\HH^\mot_n)^{-1}e_{1,\ell} \sv \,\HH^\mot_n \} ; z_1 \big] \, \GG_1[e_{1,0} ,\{e_{1,\ell}\};z_1]
\label{svalt.03}
\eeq
upon inserting $1= ( \sv \,\HH^\mot_n) (\sv \,\HH^\mot_n)^{-1}$ between any pair of braid generators $e_{1,i}$ and using the consequence $(\sv \,\HH^\mot_n)^{-1}\,e_{1,0}\, \sv \,\HH^\mot_n = e_{1,0}$ of the fact that $e_{1,0}$ commutes with both zeta generators (see Section \ref{sec:6.2}) and the braid generators in all the constituents $\sv \,\GG^\mot_2,\ldots, \sv \,\GG^\mot_n$ of $\sv \,\HH^\mot_n$. The same type of manipulations were used in Section \ref{sec:3.1} to attain the alternative form (\ref{outprf.04}) of the motivic coaction.

The leftover task in the present proof is to match the form (\ref{svalt.03}) of the Theorem with the established formulation (\ref{svalt.01}) of the single-valued map. This matching can be done at the level of the letters by showing
\beq
 (\sv \, Z^\mot_\ell) \, e_{1,\ell} \,  (\sv \, Z^\mot_\ell)^{-1} = (\sv \,\HH^\mot_n)^{-1} \, e_{1,\ell} \, \sv \,\HH^\mot_n
 \label{svalt.04}
\eeq
for all of $\ell = 2,\ldots,n{+}1$, or equivalently
\beq
(\sv \, \MM^\mot)^{-1} \, e_{1,\ell} \, \sv\, \MM^\mot = \sv\, (\GG^\mot_n\cdots\GG^\mot_2 Z^\mot_\ell)\, e_{1,\ell}\, \sv\,(\GG^\mot_n\cdots\GG^\mot_2 Z^\mot_\ell)^{-1}\, .
 \label{svalt.05}
\eeq
This in turn is a consequence of (\ref{keyeq}) (which was proven as Lemma \ref{lem:DIcon} in Section \ref{sec:6.2})
under the single-valued map\footnote{As a map from ${\cal P}^\dR$ to ${\cal P}^\mot$, see \cite{Brown:2013gia, Francislecture} and Remark \ref{svdrm}, followed by the period map.}
$(\MM^\dR, \,\GG^\dR_i,\,  Z^\dR_\ell) \rightarrow (\sv \, \MM^\mot, \, \sv \, \GG^\mot_i,\, \sv \, Z^\mot_\ell)$. 
More precisely, deducing (\ref{svalt.05}) from Lemma \ref{lem:DIcon} makes use of the fact that the single-valued map preserves functional identities among motivic MZVs and MPLs and their de Rham projections \cite{Brown:2013gia, DelDuca:2016lad, Brown:2018omk}.

\appendix

\section{Identities from free Lie algebras}
\label{app:Lie}

For some set ${\cal S}$, consider the set of words ${\cal W} = {\cal S}^\times$. A word, $A \in {\cal W}$ is an ordered set $A = a_1 a_2 \cdots a_m$ of elements $a_i$ of ${\cal S}$. Taking arbitrary finite linear combinations of words gives the \emph{free algebra} $k\!\left<{\cal S} \right>$ over ${\cal S}$, with product given by concatenation of words, which we denote $AB$, for two words $A$ and $B$. A natural inner product on $k\!\left<{\cal S} \right>$ is given by
\beq\label{lie:inner}
(A,B) = \left\{\begin{matrix} 1 & \qquad \text{if}~A = B \, , \\ 0 & \qquad \text{if}~A\neq B \, .\end{matrix} \right.
\eeq
We can also define a second product on $k\!\left<{\cal S} \right>$: the \emph{shuffle product}. For a letter $a \in {\cal S}$ and $\emptyset$ the empty word, define $a\shuffle \emptyset = \emptyset \shuffle a = a$. For some words $A,B$ and letters $a,b$, the shuffle product is inductively defined by
\beq\label{lie:shuffle}
(Aa) \shuffle (Bb) = (A \shuffle Bb) a + (Aa\shuffle B) b\, ,
\eeq
for example
\beq
aa' \shuffle bb' = aa'bb'+aba'b+baa'b' + abb'a'+bab'a' + bb'aa' \, .
\eeq
The shuffle product together with the deconcatenation coproduct\footnote{I.e.\ the coproduct $\delta A = \sum_{A=BC} B\otimes C$ as opposed to the de-shuffle coproduct in (\ref{deshuffagain}) below.} gives $k\!\left<{\cal S} \right>$ the structure of a Hopf algebra, with antipode defined by
\beq\label{lie:antipode1}
\alpha(A) = (-1)^{|A|} A^t \, ,
\eeq
where $|A|$ is the length of the word $A$ and $A^t$ is the reversed word. This implies the following identities for $A\neq \emptyset$
\beq\label{lie:antipode2}
\sum_{B,C} (A,BC) \alpha(B) \shuffle C = 0\, ,\qquad \sum_{B,C} (A,BC) B \shuffle \alpha(C) = 0 \, .
\eeq
We refer to these as the \emph{antipode identities}.

\subsection{Lie polynomials} 

A polynomial $P \in k \!\left<{\cal S} \right>$ is a \emph{Lie polynomial} if it can be written as a sum of Lie monomials, given by total bracketings of letters by the commutator
\beq
[a,b] = ab - ba \, .
\eeq
For instance, $P = [a,b] + [[a,b],b]$ is a Lie polynomial. Write
\beq
\ell[a,b,c,\ldots,d] = [[\ldots[[a,b],c],\ldots],d]
\eeq
for the total \emph{left bracketing} of some letters $a,b,c,\ldots,d$. For a letter, $a$, and a word, $A$, the left bracketing of the word $aA$ is given by
\beq\label{lie:left}
\ell[a, A] = \sum_{B,C} (A,\alpha(B)\shuffle C)\, BaC \, .
\eeq
Indeed, by \eqref{lie:shuffle}, we have that
\beq
\sum_{B,C} (bA,\alpha(B)\shuffle C)\, BaC = \sum_{B,C} (A,\alpha(B)\shuffle C)\, B[a,b]C \, ,
\eeq
and so \eqref{lie:left} follows by induction. $P$ is a Lie polynomial if and only if \cite{Reutenauer}
\beq
(A\shuffle B, P) = 0 \, , \ \ \ \ A,B \neq \emptyset \, .
\eeq
This condition can equivalently be expressed as
\beq
\delta_\shuffle P = P\otimes 1 + 1 \otimes P\, ,
\label{deshuffagain}
\eeq
where $\delta_\shuffle A = \sum_{B,C} (A,B\shuffle C) B\otimes C$ is the de-shuffle coproduct. 

\subsection{Infinite series} 

Consider formal infinite series of the form
\beq
\Psi = 1 + \sum_A \psi(A) A\, ,
\eeq
where the sum is over all non-empty words $A \in {\cal W}$ and $\psi(A)$ are some coefficients. Such a series is called \emph{group-like} if the coefficients satisfy
\beq
\psi(A)\psi(B) = \sum_C (C,A\shuffle B) \psi(C)
\eeq
for all $A,B \in {\cal W}$.
In other words, using the de-shuffle coproduct, a group-like series satisfies
\beq
\delta_\shuffle \Psi = \Psi \otimes \Psi \, .
\eeq
For such a group-like series, $\Psi$, its inverse series (with respect to the concatenation product)~is
\beq\label{lie:phi-inv}
\Psi^{-1} = 1 + \sum_A \psi(A) \alpha(A) \, ,
\eeq
where the antipode $\alpha(A)$ is given by (\ref{lie:antipode1}).
Indeed, the antipode identity, \eqref{lie:antipode2}, implies that $\Psi\Psi^{-1} = \Psi^{-1}\Psi = 1$. Moreover, for some letter $a$, the conjugation of $a$ by a group-like series is itself an infinite Lie series,
\beq\label{lie:phi-conj}
\Psi^{-1} a \Psi = \sum_{B,C} \psi \big(\alpha(B)\shuffle C \big) B a C = \sum_A \psi(A) \ell[a, A]\,,
\eeq
which follows from \eqref{lie:left}.

\section{The pure braid group and the main theorem}
\label{sec:D}

We use the action of braids on MPLs at a key step, Lemma \ref{lem:bigphiprod}, in the proof of our first main theorem. There, we use the series $\BB^\dR(\sigma_{(ab)})$ given by (\ref{eq:braid:cycle}) which implements the braid $\sigma_{(a,b)} = \sigma_{a,a+1}\cdots \sigma_{b-1,b}$, which induces a cyclic permutation $(ab)$ on the indices. Other braids also implement the same cyclic permutation $(ab)$ of the indices. These braids differ from $\sigma_{(a,b)}$ by elements of the \emph{pure braid group}. However, it can be seen that acting with elements of the pure braid group does \emph{not} lead to different formulas. Indeed, consider the pure braid $\sigma = \sigma_{i,i+1} \sigma_{i,i+1}$. The series that implements this pure braid is
\begin{align}
\BB^\dR(\sigma) = \left(\tau_{\sigma_{i,i+1}} \BB^\dR(\sigma_{i,i+1})\right) \BB^\dR(\sigma_{i,i+1})
&= \Phi^\dR\left(\sum_{j=0}^{i-1} e_{j,i},e_{i,i+1}\right) \Phi^\dR\left(e_{i,i+1},\sum_{j=0}^{i-1} e_{j,i+1}\right)  \\ 
&\quad \times \Phi^\dR\left(\sum_{j=0}^{i-1} e_{j,i+1},e_{i,i+1}\right) \Phi^\dR\left(e_{i,i+1},\sum_{j=0}^{i-1} e_{j,i}\right) \, .
\notag
\end{align}
However, using $\Phi(e_0,e_1)\Phi(e_1,e_0) = 1$, we see that
\beq
\BB^\dR(\sigma) = 1 \, .
\eeq
In fact, all pure braids can be generated from braids of the form
\beq
\sigma_{((a,b))} = \sigma_{(ab)}\sigma_{(ba)} =  (\sigma_{a,a+1}\cdots \sigma_{b-1,b}) (\sigma_{b-1,b} \cdots \sigma_{a,a+1})\, ,
\eeq
which `wraps' strand $a$ clockwise around strands $a{+}1,\ldots,b$ and then returns it to its original position. By the same type of argument used in the proof of Lemma \ref{lem:braid:cycle}, we find
\beq
\BB^\dR(\sigma_{(ba)}) = \prod_{i=a}^{b-1} \Phi^\dR\left(  \sum_{j=0}^{i-1} e_{j,b} \, , \, e_{b,i} \right) \Phi^\dR\left( e_{b,i} \, , \, \sum_{j=0}^{i-1} e_{j,i} \right) \, ,
\eeq
where the order of multiplication is left-to-right from $i=a$ to $i=b-1$. Combining this with Lemma \ref{lem:braid:cycle}, we find that
\beq
\BB^\dR(\sigma_{((a,b))}) = \tau_{\sigma_{(ab)}}\BB^\dR(\sigma_{(ba)}) \BB^\dR(\sigma_{(ab)}) = 1\, ,
\eeq
where we again use $\Phi(e_0,e_1)\Phi(e_1,e_0) = 1$. In other words, in the proof of Lemma \ref{lem:bigphiprod}, we can use \emph{any} braid that corresponds to the cyclic permutation $(1,2,\ldots,\ell{-}1)$ and find the same result. Hence, the expression (\ref{keyrst}) for the limit of ${\cal G}_n^\dR$ and the resulting action (\ref{eq:adaction}) of zeta generators do not depend on our choice $\sigma_{(1,\ell-1)} = \sigma_{1,2} \sigma_{2,3}\ldots \sigma_{\ell-2,\ell-1}$ of cycle braid.

\newpage

\providecommand{\href}[2]{#2}\begingroup\raggedright\endgroup


\end{document}